\documentclass[11pt]{article}

\usepackage{mathpazo}
\usepackage{amsfonts}
\usepackage{amsmath}
\usepackage{graphicx}
\usepackage{latexsym}
\usepackage{mathabx}

\usepackage[margin=1.05in]{geometry}
  
\usepackage[T1]{fontenc}
\usepackage{times}
\usepackage{color,graphicx}
\usepackage{array}
\usepackage{enumerate}
\usepackage{amsmath}
\usepackage{amssymb}
\usepackage{amsthm}
\usepackage{pgfplots}
\usepackage{pgf}
\usepackage{tikz}
\usetikzlibrary{patterns}
\usepgfplotslibrary{patchplots} 
\usetikzlibrary{pgfplots.patchplots} 
\pgfplotsset{width=9cm,compat=1.5.1}

\usepackage[english]{babel}
\usepackage{amsthm}
\newtheorem{theorem}{Theorem}
\newtheorem{conjecture}[theorem]{Conjecture}

\newtheorem{lemma}[theorem]{Lemma}
\newtheorem{corollary}[theorem]{Corollary}
\newtheorem{remark}[theorem]{Remark}
\newtheorem{proposition}[theorem]{Proposition}

\newtheorem{example}[theorem]{Example} 

\usepackage{xcolor}
\usepackage{makeidx}
\usepackage[colorlinks=true,linkcolor=blue,anchorcolor=blue,citecolor=red,urlcolor=magenta]{hyperref}
\usepackage[alphabetic,backrefs]{amsrefs}

\begin{document}
\title{Quantum walks on join graphs}

\author{
	Steve Kirkland,\textsuperscript{\!\!1} Hermie Monterde \textsuperscript{\!\!1}
}

\maketitle










\begin{abstract}
The join $X\vee Y$ of two graphs $X$ and $Y$ is the graph obtained by joining each vertex of $X$ to each vertex of $Y$. We explore the behaviour of a continuous quantum walk on a weighted join graph having the adjacency matrix or Laplacian matrix as its associated Hamiltonian. We characterize strong cospectrality, periodicity and perfect state transfer (PST) in a join graph. We also determine conditions in which strong cospectrality, periodicity and PST are preserved in the join. Under certain conditions, we show that there are graphs with no PST that exhibits PST when joined by another graph. This suggests that the join operation is promising in producing new graphs with PST. Moreover, for a periodic vertex in $X$ and $X\vee Y$, we give an expression that relates its minimum periods in $X$ and $X\vee Y$. While the join operation need not preserve periodicity and PST, we show that $\big| |U_M(X\vee Y,t)_{u,v}|-|U_M(X,t)_{u,v}| \big|\leq \frac{2}{|V(X)|}$ for all vertices $u$ and $v$ of $X$, where $U_M(X\vee Y,t)$ and $U_M(X,t)$ denote the transition matrices of $X\vee Y$ and $X$ respectively relative to either the adjacency or Laplacian matrix. We demonstrate that the bound $\frac{2}{|V(X)|}$ is tight for infinite families of graphs. 
\end{abstract}

	


\addtocounter{footnote}{1}
\footnotetext{Department of Mathematics, University of Manitoba, Winnipeg, MB, Canada R3T 2N2}


\section{Introduction}\label{secINTRO}

The graphs $K_2$, $C_4$ $P_3$, and $K_{4}\backslash e$ are well-known examples of small graphs that admit adjacency or Laplacian perfect state transfer. Kirkland et al.\ noticed that in these small graphs, the vertices involved in perfect state transfer  share the same neighbours, an observation that prompted them to examine state transfer between twins \cite{Kirkland2023}. However, it can also be observed that these small graphs are in fact join graphs. This motivates our investigation on quantum walks on graphs built using the join operation.

The first infinite family of join graphs revealed to admit Laplacian perfect state transfer was $K_{n}\backslash e:=O_2\vee K_{n-2}$ (the complete graph on $n$ vertices minus an edge) with $n\equiv 0$ (mod 4) \cite{Bose2008}. Motivated by this result,  Angeles-Canul et al.\ gave sufficient conditions for adjacency perfect state transfer to occur between the vertices of $X\in\{O_2,K_2\}$ in the unweighted graph $X\vee Y$, where $Y$ is a regular graph \cite{Angeles-Canul2010}. Such a graph is called a double cone  on $Y$, and the vertices of $X\in\{O_2,K_2\}$ are called the apexes of the double cone. They also determined sufficient conditions such that adjacency perfect state transfer in $X$ is preserved under joins of copies of $X$. In a subsequent paper, Angeles-Canul et al.\ investigated perfect state transfer in weighted join graphs, and found that the apexes of a double cone on a $k$-regular graph admit adjacency perfect state transfer by appropriate choice of weights of an edge between the apexes and/or loops on the apexes \cite{Angeles-Canul2009}. More recently, Kirkland et al.\ fully characterized unweighted double cones that admit adjacency perfect state transfer between apexes \cite{Kirkland2023}. For the Laplacian case, Alvir et al.\ showed that the apexes of unweighted $O_2\vee Y$ admit perfect state transfer if and only if $|V(Y)|\equiv 2$ (mod 4), while the apexes of unweighted $K_2\vee Y$ do not admit perfect state transfer \cite{Alvir2016}. Joins have also been investigated in graphs with well-structured eigenbases \cite{Johnston2017,mclaren2023weak}, and in the contexts of fractional revival \cite{Monterde2023a,Chan2021}, sedentariness \cite{Monterde2023} and strong cospectrality \cite{Monterde2022}.

Despite its widespread presence in literature, state transfer on join graphs remains largely unexplored. In this paper, we provide a systematic study of quantum walks on weighted join graphs having the adjacency and Laplacian matrices as their associated Hamiltonian. In Section \ref{secJoin}, we provide known results about transition matrices and eigenvalue supports of vertices in join graphs. In Section \ref{secPer}, we characterize periodic vertices in a join (Theorem \ref{joinper}), and determine the conditions in which periodicity of a vertex in $X$ and $X\vee Y$ are equivalent (Corollaries \ref{joinperpreserve} and \ref{perequiv}). We also use the join operation to provide infinite families of graphs that are not integral (resp., Laplacian integral) whereby each member graph contains periodic vertices (Corollary \ref{perjoin1}). We devote Section \ref{secMinper} to finding a relationship between the minimum periods $\rho_X$ and $\rho_{X\vee Y}$ of a vertex that is periodic in both $X$ and $X\vee Y$, resp.\ (Theorems \ref{perjoin} and \ref{perjoinn}). We find that $\rho_{X\vee Y}$ is  a rational multiple of $\rho_{X}$, and if we add the extra hypothesis that $X$ is disconnected, then $\rho_{X\vee Y}$ is an integer multiple of $\rho_{X}$. Section \ref{secSc} provides a characterization of strong cospectrality in joins (Theorem \ref{scj}). It turns out that strong cospectrality between two vertices of $X$ in $X\vee Y$ requires strong cospectrality in $X$, except for the case when $X$ is an empty graph on two vertices. In fact, for the unweighted case, strong cospectrality in $X$ and $X\vee Y$ are equivalent with a few exceptions (Corollary \ref{scj}). In Section \ref{secPst}, we characterize perfect state transfer in joins (Theorems \ref{pstL} and \ref{pstA}). Unlike strong cospectrality, perfect state transfer between two vertices of $X$ in $X\vee Y$ does not require perfect state transfer between them in $X$. This motivates us in Section \ref{secPst1} to determine necessary and sufficient conditions such that perfect state transfer in $X$ is preserved in $X\vee Y$ (Corollaries \ref{pstLc2} and \ref{adjpstinX}), and in Section \ref{secPst2} to determine conditions such that perfect state transfer occurs in $X\vee Y$ given that it does not occur in $X$ (Theorems \ref{oddst1} and \ref{oddst}, and Corollary \ref{pstAc5}). In Section \ref{secFam}, we characterize strong cospectrality and PST in self-joins and iterated join graphs. The latter result generalizes a previously known characterization of perfect state transfer in threshold graphs. While the join operation does not preserve periodicity and PST in general, we show in Section \ref{secBounds} that $\big| |U_M(X\vee Y,t)_{u,v}|-|U_M(X,t)_{u,v}| \big|\leq \frac{2}{|V(X)|}$ for all $t$ and for all vertices $u$ and $v$ of $X$ (Corollaries \ref{boundL1} and \ref{boundA1}). Consequently, the behaviour of vertices in $X$ in the quantum walk on $X\vee Y$ mimics the behaviour of the quantum walk on $X$ as $|V(X)|\rightarrow\infty$ (Corollary \ref{mimic}). We also show that this upper bound is tight for infinite families of graphs (Examples \ref{exa} and \ref{exaa}). Finally, open problems are presented in Section \ref{secFw}.

\section{Preliminaries}

Throughout, we assume that a graph $X$ is weighted and undirected  with possible loops but no multiple edges. We denote the vertex set of $X$ by $V(X)$, and we allow the edges of $X$ to have positive real weights. We say that $X$ is \textit{simple} if $X$ has no loops, and $X$ is \textit{unweighted} if all edges of $X$ have weight one. For $u\in V(X)$, we denote the the characteristic vector of $u$ as $\textbf{e}_u$, which is a vector with a $1$ on the entry indexed by $u$ and zeros elsewhere. The all-ones vector of order $n$, the zero vector of order $n$, the $m\times n$ all-ones matrix, and the $n\times n$ identity matrix are denoted by $\textbf{1}_n$, $\textbf{0}_n$, $\textbf{J}_{m,n}$ and $I_n$, respectively. If $m=n$, then we write $\textbf{J}_{m,n}$ as $\textbf{J}_n$, and if the context is clear, then we simply write these matrices as $\textbf{1}$, $\textbf{0}$, $\textbf{J}$ and $I$, respectively. We also represent the transpose of the matrix $H$ by $H^T$, and the characteristic polynomial of $H$ in the variable $t$ by $\phi(H,t)$. Lastly, we denote the simple unweighted empty, cycle, complete, and path graphs on $n$ vertices as $O_n$, $C_n$, $K_n$, and $P_n$, respectively. We also denote the hypercube on $2^p$ vertices by $Q_p$, and the cocktail party graph on $m$ vertices by $CP(m)$, where $m$ is even.

The adjacency matrix $A(X)$ of $X$ is the matrix defined entrywise as
\begin{equation*}
A(X)_{u,v}=
\begin{cases}
 \omega_{u,v}, &\text{if $u$ and $v$ are adjacent}\\
 0, &\text{otherwise},
\end{cases}
\end{equation*}
where $\omega_{u,v}$ is the weight of the edge $(u,v)$. The degree matrix $D(X)$ of $X$ is the diagonal matrix of vertex degrees of $X$, where $\operatorname{deg}(u)=2\omega_{u,u}+\sum_{j\neq u}\omega_{u,j}$ for each $u\in V(X)$. The Laplacian matrix of $X$ is the matrix $L(X)=D(X)-A(X)$. We use $M(X)$ to denote $A(X)$ or $L(X)$. If the context is clear, then we write $M(X)$, $A(X)$, $L(X)$ and $D(X)$ resp.\ as $M$, $A$, $L$ and $D$. We say that $X$ is \textit{integral} (resp., \textit{Laplacian integral}) if all eigenvalues of $A(X)$ (resp., $L(X)$) are integers. We say that $X$ is \textit{weighted $k$-regular} if $\operatorname{deg}(u)-\omega_{u,u}$ is a constant $k$ for each $u\in V(X)$, i.e., the row sums of $A(X)$ are a constant $k$. If $X$ is simple, then $X$ is weighted $k$-regular if and only if $D(X)=kI$.

A \textit{(continuous-time) quantum walk} on $X$ with respect to $H$ is determined by the unitary matrix
\begin{equation}
\label{M}
U_H(X,t)=e^{itH}.
\end{equation}
The matrix $H$ is called the \textit{Hamiltonian} of the quantum walk. Typically, $H$ is taken to be $A$ or $L$, but any Hermitian matrix $H$ that respects the adjacencies of $X$ works in general. Since $M\in\{A,L\}$ is real symmetric, $M$ admits a spectral decomposition
\begin{equation}
\label{specdec}
M=\sum_{j}\lambda_jE_j,
\end{equation}
where the $\lambda_j$'s are the distinct real eigenvalues of $M$ and each $E_j$ is the orthogonal projection matrix onto the eigenspace associated with $\lambda_j$. If the eigenvalues are not indexed, then we also denote by $E_{\lambda}$ the orthogonal projection matrix corresponding to the eigenvalue $\lambda$ of $M$. Now, (\ref{specdec}) allows us to write (\ref{M}) as
\begin{equation*}
U_M(X,t)=\sum_{j}e^{it\lambda_j}E_j.
\end{equation*}

We say that \textit{perfect state transfer} (PST) occurs between two vertices $u$ and $v$ if $|U_M(X,\tau)_{u,v}|=1$ for some time $\tau$. The minimum $\tau>0$ such that $|U_M(X,\tau)_{u,v}|=1$ is called the \textit{minimum PST time} between $u$ and $v$. We say that $u$ is \textit{periodic} if $|U_M(X,\tau)_{u,u}|=1$ for some time $\tau$. The minimum $\tau>0$ such that $|U_M(X,\tau)_{u,u}|=1$ is called the \textit{minimum period} of $u$. These properties depend on the matrix $M$, so we sometimes say adjacency PST (resp., periodicity) when $M=A$; similar language applies when $M=L$. Further, if $X$ is simple, weighted and $k$-regular, then $U_L(X,t)=e^{itL}=e^{it(kI-A)}=e^{itk}e^{-itA}=e^{itk}U_A(X,-t)$. Hence, $\left|U_L(X,\tau)_{u,v}\right|^2=\left|U_A(X,\tau)_{u,v}\right|^2$ for any $u,v\in V(X)$, so the quantum walks on $A$ and $L$ exhibit the same  state transfer properties.

\section{Joins}\label{secJoin}

The \textit{join} $X\vee Y$ of weighted graphs $X$ and $Y$ is the graph obtained by joining every vertex of $X$ with every vertex of $Y$ with an edge of weight one. Unless otherwise stated, we make the following  assumptions:
\begin{itemize}
    \item when dealing with $M=L,$ we take $X$ and $Y$ to be simple;
    \item when dealing with $M=A,$  we take $X$ and $Y$ to be weighted $k$- and $\ell$-regular resp.\ (possibly with loops).  
\end{itemize}

Denote the transition matrices of $X\vee Y$ and $X$ by $U_M(X\vee Y,t)$ and $U_M(X,t)$ respectively. The transition matrix of a join with respect to $L$ and $A$ are known (see \cite[Fact 8]{Alvir2016} and \cite[Lemma 4.3.1]{Coutinho2021}, respectively). Throughout, $m$ and $n$ denote the number of vertices of $X$ and $Y$, respectively.

\begin{lemma}
\label{ljoin}
If $u,v\in V(X)$ and $w\in V(Y)$, then
\begin{equation}
\label{offL}
U_L(X\vee Y,t)_{u,w}=\frac{1}{m+n}\left(1-e^{it(m+n)}\right)
\end{equation}
and
\begin{equation}
\label{notoffL}
U_L(X\vee Y,t)_{u,v}=e^{itn}U_L(X,t)_{u,v}+\frac{1}{m+n}+\frac{ne^{it(m+n)}}{m(m+n)}-\frac{e^{itn}}{m}.
\end{equation}
\end{lemma}

\begin{lemma}
\label{aqjoin}
Let $D=(k-\ell)^2+4mn$ and $\lambda^{\pm}=\frac{1}{2}\left(k+\ell\pm\sqrt{D}\right)$. If $u,v\in V(X)$ and $w\in V(Y)$, then
\begin{equation}
\label{offA}
U_A(X\vee Y,t)_{u,w}=\frac{1}{\sqrt{D}}(e^{it\lambda^+}-e^{it\lambda^-})
\end{equation}
and 
\begin{equation}
\label{notoffA}
U_A(X\vee Y,t)_{u,v}=U_A(X,t)_{u,v}+\frac{e^{it\lambda^+}(k-\lambda^-)}{m\sqrt{D}}-\frac{e^{it\lambda^-}(k-\lambda^+)}{m\sqrt{D}}-\frac{e^{it k}}{m}.
\end{equation}
\end{lemma}

The \textit{eigenvalue support} of $u$ with respect to $M$ is the set
\begin{equation*}
\sigma_u(M)=\{\lambda_j:E_j\textbf{e}_u\neq \textbf{0}\}.
\end{equation*}
The following result determines the elements of the eigenvalue support of a vertex in a join graph. In what follows, $\sigma_u(M(X))$ and $\sigma_u(M)$ denote the eigenvalue supports of vertex $u$ in $X$ and $X\vee Y$, respectively.

\begin{lemma}
\label{supp}

Let $u\in V(X)$, and define $\mathcal{S}:=\{\lambda+n: \lambda\in\sigma_u(L(X))\backslash\{0\}\}$.
\begin{enumerate}
\item We have $\sigma_u(L)=\mathcal{S}\cup \mathcal{R}$, where $\mathcal{R}=\{0,m+n\}$ if $X$ is connected and $\mathcal{R}=\{0,m+n,n\}$ otherwise. 
\item We have $\sigma_u(A)= \sigma_u(A(X))\backslash\{k\} \cup \mathcal{R}$, where $\mathcal{R}=\{\lambda^{\pm}\}$ if $X$ is connected and $\mathcal{R}=\{\lambda^{\pm},k\}$ otherwise.
\end{enumerate}
\end{lemma}

\section{The ratio condition}\label{secPer}

A set $S\subseteq\mathbb{R}$ with at least two elements satisfies the \textit{ratio condition} if $\frac{\lambda-\theta}{\mu-\eta}\in\mathbb{Q}$ for all $\lambda,\theta,\mu,\eta\in S$ with $\mu\neq \eta$. A vertex $u$ is periodic if and only if $\sigma_u(M)$ satisfies the ratio condition \cite[Corollary 8.3.1]{Coutinho2021}. Combining this fact with Lemma \ref{supp} yields a characterization of periodicity in joins.

\begin{theorem}
\label{joinper}
Let $X$ and $Y$ be weighted graphs and $u\in V(X)$.
\begin{enumerate}
\item Vertex $u$ is Laplacian periodic in $X\vee Y$ if and only if all eigenvalues in $\sigma_u(L(X))$ are rational. Moreover, if $u$ is Laplacian periodic in $X\vee Y$, then it is Laplacian periodic in $X$.
\item If $X$ is connected (resp., disconnected), then $u$ is adjacency periodic if and only if $\frac{\lambda^+-\lambda}{\sqrt{D}}\in\mathbb{Q}$ for all $\lambda\in \sigma_u(A(X))\backslash\{k\}$ (resp., $\lambda\in \sigma_u(A(X))$). If we assume further that $k,\ell,\sqrt{D}\in\mathbb{Q}$, then $u$ is adjacency periodic in $X\vee Y$ if and only if all eigenvalues in $\sigma_u(A(X))$ are rational.

\end{enumerate}
\end{theorem}

The following example illustrates that the converse of the last statement in Theorem \ref{joinper}(1) does not hold.

\begin{example}
Consider $X\vee O_1$, where $X$ is a weighted $K_2$ with positive edge weight $\eta\notin\mathbb{Q}$. Then $u\in V(X)$ is periodic in $X$ and $\sigma_u(L)=\{0,3,1+2\eta\}$. Since $\frac{1+2\eta-0}{3-0}\notin\mathbb{Q}$, we get that $u$ is not periodic in $X\vee O_1$.
\end{example}

If $\phi(M(X),t)\in\mathbb{Z}[t]$, then we obtain another characterization of periodicity \cite[Corollary 7.6.2]{Coutinho2021}.

\begin{lemma}
\label{per}
Let $S$ be a set of real algebraic integers that is closed under taking algebraic conjugates. If $S$ satisfies the ratio condition, then either (i) $S\subseteq\mathbb{Z}$ or (ii) each $\lambda_j$ in $S$ is a quadratic integer of the form $\lambda_j=\frac{1}{2}(a+b_j\sqrt{\Delta})$, where $a,b_j,\Delta\in\mathbb{Z}$ and $\Delta>1$ is square-free. In particular, if $X$ is a weighted graph with possible loops such that $\phi(M(X),t)\in\mathbb{Z}[t]$, then vertex $u$ is periodic in $X$ if and only if either (i) or (ii) holds for $S=\sigma_u(M(X))$. In this case, $\rho=2\pi/g\sqrt{\Delta}$, where each $\lambda_j\in\sigma_u(M)$, $g=\operatorname{gcd}\{(\lambda_1-\lambda_j)/\sqrt{\Delta}\}$ and $\Delta=1$ if (i) holds. 
\end{lemma}

We now show that the converse of the last statement in Theorem \ref{joinper}(1) holds when $\phi(L(X),t)\in\mathbb{Z}[t]$.

\begin{corollary}
\label{joinperpreserve}
If $\phi(L(X),t)\in\mathbb{Z}[t]$, then the following statements are equivalent.
\begin{enumerate}
\item Vertex $u$ is Laplacian periodic in $X$.
\item Vertex $u$ is Laplacian periodic in $X\vee Y$ for any simple positively weighted graph $Y$.
\item The eigenvalues in $\sigma_u(L(X))$ are all integers.
\end{enumerate}
\end{corollary}

\begin{proof}
Lemma \ref{supp}(1) and the ratio condition yields (3) implies (2). The second statement of Theorem \ref{joinper}(1) yields (2) implies (1). Lemma \ref{per}(i) and the fact that $0\in\sigma_u(L(X))$ yield (1) implies (3).
\end{proof}

Next, we have the following result concerning the adjacency matrix in relation to Theorem \ref{joinper}(2).

\begin{corollary}
\label{joinper1}
Let $u\in V(X)$, $\phi(A(X),t)\in\mathbb{Z}[t]$ and $k,\ell\in\mathbb{Z}$. Consider $D$ in Lemma \ref{aqjoin}.
\begin{enumerate}
\item If $X$ is connected, then $u$ is adjacency periodic in $X\vee Y$ if and only if either:
\begin{enumerate}
\item The eigenvalues in $\sigma_u(A(X))$ are all integers and $D$ is a perfect square (so that $\lambda^{\pm}\in\mathbb{Z}$). In this case, $u$ is also adjacency periodic in $X$.
\item Each $\lambda_j\in\sigma_u(A(X))\backslash\{k\}$ is a quadratic integer of the form $\frac{1}{2}(k+\ell+b_j\sqrt{\Delta})$, where $\Delta>1$ is square-free, $\sqrt{D/\Delta}\in\mathbb{Z}$ and $\sqrt{D/\Delta}>b_j$ for each $j$. In this case, $u$ is not periodic in $X$.
\end{enumerate}
\item If $X$ is disconnected, then $u$ is adjacency periodic in $X\vee Y$ if and only if condition (1a) holds.
\end{enumerate}
\end{corollary}

\begin{proof}
To prove (1), suppose $X$ is connected. By Lemma \ref{supp}(2), $\sigma_u(A)=\sigma_u(A(X))\backslash\{k\}\cup \{\lambda^{\pm}\}$, and $u$ is periodic in $X\vee Y$ if and only if $\sigma_u(A)$ satisfies the ratio condition. Since $k\in\mathbb{Z}$, the latter condition implies that $\sigma_u(A(X))\backslash\{k\}$ satisfies the ratio condition, and as $\phi(A(X),t)\in\mathbb{Z}[t]$, we get two cases from Lemma \ref{per}. First, let $\sigma_u(A(X))\backslash\{k\}\subseteq \mathbb{Z}$. Since $\lambda^+=\frac{1}{2}(k+\ell+\sqrt{D})$, we get $\frac{\lambda^+-\lambda}{\sqrt{D}}\in\mathbb{Q}$ if and only if $\sqrt{D}\in\mathbb{Z}$. Combining this with Theorem \ref{joinper}(2) yields (1a). Next, suppose each $\lambda_j\in\sigma_u(A(X))\backslash\{k\}$ is of the form $\frac{1}{2}(a+b_j\sqrt{\Delta})$ so that $\lambda^+-\lambda_j=\frac{1}{2}\left((k+\ell-a)+(\sqrt{D}-b_j\sqrt{\Delta})\right)$. Then $\frac{\lambda^+-\lambda_j}{\sqrt{D}}\in\mathbb{Q}$ if and only if $k+\ell-a=0$ and $\sqrt{\Delta/D}\in\mathbb{Q}$. Now, Perron-Frobenius Theorem yields $\lambda^+=\frac{1}{2}(a+\sqrt{D})$ as the largest eigenvalue of $A(X\vee Y)$. Thus, $\sqrt{D}>\sqrt{\Delta}$. Since $\sigma_u(A)$ satisfies the ratio condition, we get $\sqrt{D}=b\sqrt{\Delta}$ for some integer $b$, and so $\sqrt{D/\Delta}=b>b_j$ for each $j$. This establishes (1b).

To prove (2), suppose $X$ is disconnected. In this case, $k\in\sigma_u(A(X))$ is an integer, and so we apply Lemma \ref{joinper}(2) to conclude that $u$ is periodic in $X\vee Y$ if and only if $\sigma_u(A(X))\subseteq\mathbb{Z}$ and $\sqrt{D}\in\mathbb{Q}$.
\end{proof}

We now state an analogue of Corollary \ref{joinperpreserve}, which is an immediate consequence of Corollary \ref{joinper1}.

\begin{corollary}
\label{perequiv}
Let $\phi(A(X\vee Y),t)\in\mathbb{Z}[t]$ and $k,\ell\in\mathbb{Z}$. If $u$ is adjacency periodic in $X$, then $u$ is adjacency periodic in $X\vee Y$ if and only if $D$ is a perfect square. Moreover, if $D$ is a perfect square, then the following are equivalent.
\begin{enumerate}
\item Vertex $u$ is adjacency periodic in $X$.
\item Vertex $u$ is adjacency periodic in $X\vee Y$ for any weighted $\ell$-regular graph $Y$.
\item All eigenvalues in $\sigma_u(A(X))$ are  integers.
\end{enumerate}
\end{corollary}

\begin{remark}
\label{joinperpreserve1}
For $L$, our assumption on $X$ and $Y$ implies that the join operation preserves periodicity, and periodicity in the join is inherited by the underlying graph. However, this does not hold for $A$ given our assumption on $X$ and $Y$. One can have $\sigma_u(A(X))\subseteq\mathbb{Z}$ but $\sqrt{D}\notin\mathbb{Z}$, and so $u$ is periodic in $X$ but not in $X\vee Y$. Conversely, periodicity in $X\vee Y$ is not necessarily inherited by $X$ by Corollary \ref{joinper}. 
\end{remark}

We say that $X$ is \textit{periodic} if there is a $\tau>0$ such that $U(\tau)=\gamma I$ for some unit $\gamma\in\mathbb{C}$. Equivalently, each $u\in V(X)$ is periodic at time $\tau$. We characterize periodic join graphs under mild conditions.

\begin{theorem}
\label{Xper}
Let $X$ and $Y$ be weighted graphs and  $\phi(M(X\vee Y),t)\in\mathbb{Z}[t]$.
\begin{enumerate}
\item If $X$ and $Y$ are simple, then $X\vee Y$ is Laplacian periodic if and only if $X\vee Y$ is Laplacian integral.
\item If $X$ and $Y$ are regular, then $X\vee Y$ is adjacency periodic if and only if $X\vee Y$ is integral.
\end{enumerate}
\end{theorem}

\begin{proof}
Note that (1) follows from the ratio condition and the fact that $L(X\vee Y)$ is positive semidefinite with $0$ as an eigenvalue. Next, let $X$ and $Y$ be $k$- and $\ell$-regular. By way of contradiction, suppose each eigenvalue $\lambda_j$ of $A(X\vee Y)$ is a quadratic integer. Since $\lambda^{\pm}=\frac{1}{2}(k+\ell\pm\sqrt{D})$ is an eigenvalue of $A(X\vee Y)$, each $\lambda_j$ must be of the form $\frac{1}{2}(k+\ell+b_j\sqrt{\Delta})$. Without loss of generality, let $k\geq \ell>0$. For each eigenvalue $\lambda_j\neq \ell$ of $A(Y)$, we get $\lambda_j=\frac{1}{2}(k+\ell+b_j\sqrt{\Delta})\geq \ell+\frac{1}{2}b_j\sqrt{\Delta}$ for all each $j$. Similarly, $\overline{\lambda_j}\geq \ell-\frac{1}{2}b_j\sqrt{\Delta}$ for all each $j$. Thus, $\lambda_j>\ell$ or $\overline{\lambda_j}>\ell$, which cannot happen because $\ell$ is the largest eigenvalue of $A(Y)$ by the 
Perron-Frobenius Theorem. Invoking Lemma \ref{per} establishes (2).
\end{proof}

We close this section by providing infinite families of graphs that are not (Laplacian) integral but contain periodic vertices. The following result is immediate from Lemma \ref{supp} and the ratio condition.

\begin{corollary}
\label{perjoin1}
Let $X$ and $Y$ be weighted graphs, and let $u\in V(X)$.
\begin{enumerate}
\item If $X$ is Laplacian integral but $Y$ is not, then $X\vee Y$ is not Laplacian integral, the eigenvalues in $\sigma_u(L(X))$ are all integers, and $u$ is periodic in $X\vee Y$.
\item If $X$ is integral but $Y$ is not, then $X\vee Y$ is not integral. Moreover, if $k,\ell\in\mathbb{Z}$ and $D$ in Lemma \ref{aqjoin} is a perfect square, then  the eigenvalues in $\sigma_u(A(X))$ are all integers and $u$ is periodic in $X\vee Y$.
\end{enumerate}
\end{corollary}

\section{Minimum periods}\label{secMinper}

Denote the minimum periods of $u$ in $X\vee Y$ and $X$ by $\rho_{X\vee Y}$ and $\rho_X$ respectively. Motivated by Corollaries \ref{joinperpreserve} and \ref{perequiv}, we ask: if $u\in V(X)$ is periodic in $X$ and $X\vee Y$, then how are $\rho_{X\vee Y}$ and $\rho_X$ related? To answer this question, we state a result due to Kirkland et al.\ \cite[Theorem 4]{Kirkland2023}.

\begin{theorem}
\label{minperiod}
Let $u$ and $v$ be vertices in $X$, and suppose $\sigma_u(M)=\{\lambda_1,\ldots,\lambda_r\}$. If $u$ is periodic, then $\rho=\frac{2\pi q}{\lambda_1-\lambda_2}$, where $q=\operatorname{lcm}(q_2,\ldots,q_r)$ and each $q_j$ is a positive integer such that $\frac{\lambda_1-\lambda_j}{\lambda_1-\lambda_2}=\frac{p_j}{q_j}$ for some integer $p_j$ satisfying $\text{gcd}(p_j,q_j)=1$. In particular, if $r=2$, then $u$ is periodic with $\rho=\frac{2\pi}{\lambda_1-\lambda_2}$.
\end{theorem}

We now address our above question by finding a constant $c>0$ such that $\rho_{X\vee Y}=c\rho_X$.

\begin{theorem}
\label{perjoin}
Let $u$ be a non-isolated and periodic vertex in $X$ and $X\vee Y$ with $\sigma_u(L(X))=\{\lambda_1,\ldots,\lambda_r,\lambda_{r+1}\}$, where $\lambda_{r+1}=0$. Let $\lambda_{r+2}:=m$ and $\lambda_{r+3}:=-n$. 
\begin{enumerate}
\item Suppose $X$ is connected.
\begin{enumerate}
\item If $r=1$, then $\rho_{X\vee Y}=\left(\frac{\lambda_1q}{\lambda_1+n}\right)\rho_X$, where $q$ is a positive integer such that $\frac{m+n}{\lambda_1+n}=\frac{p}{q}$ for some integer $p$ with $\operatorname{gcd}(p,q)=1$. In particular, if $\lambda_1=m$, then $q=1$.
\item If $r=2$ and $\lambda_2=m$, then $\rho_{X\vee Y}=\left(\frac{\lambda_1q}{q'(\lambda_1+n)}\right)\rho_X$, where $q$ is a positive integer defined in (a) and $q'$ is a positive integer such that $m/\lambda_1=p'/q'$ for some integer $p'$ with $\operatorname{gcd}(p',q')=1$.
\item If $r\geq 2$ and $\lambda_j\neq m$ for each $j\in\{1,\ldots,r\}$, then
\begin{equation}
\label{lconL}
\rho_{X\vee Y}=\left( \frac{q_{r+2}q_{r+3}\operatorname{gcd}(R_1,q_{r+1})}{q_{r+1}\operatorname{gcd}(R_1,q_{r+2})\operatorname{gcd}(R_2,q_{r+3})}\right)\rho_X, 
\end{equation}
where $R_1=\operatorname{lcm}(q_3,\ldots,q_r)$, $R_2=\operatorname{lcm}(q_2,\ldots,q_r,q_{r+2})$, the $q_j$'s are positive integers such that $\frac{\lambda_1-\lambda_j}{\lambda_1-\lambda_2}=\frac{p_j}{q_j}$ for some integers $p_j$ with $\operatorname{gcd}(p_j,q_j)=1$ and $\lambda_1>\lambda_2$.
\item If $r\geq 3$ and $\lambda_r=m$, then (\ref{lconL}) holds with $R_1=\operatorname{lcm}(q_3,\ldots,q_{r-1})$, $R_2=\operatorname{lcm}(q_3,\ldots,q_{r-1},q_{r+2})$.
\end{enumerate}
\item Suppose $X$ is disconnected.
\begin{enumerate}
\item If $r=1$, then $\rho_{X\vee Y}=q\rho_X$, where $q=\operatorname{lcm}(q_{3},q_{4})$ and the $q_j$'s are positive integers such that $\frac{\lambda_1-\lambda_j}{\lambda_1-\lambda_2}=\frac{p_j}{q_j}$ for some  $p_j$ with $\operatorname{gcd}(p_j,q_j)=1$. In particular, if $\lambda_1=m$, then $q=\frac{m}{\operatorname{gcd}(m,n)}$. 
\item If $r=2$ and $\lambda_2=m$, then $\rho_{X\vee Y}=(q/q_4)\rho_X$, where $q=\operatorname{lcm}(q_{4},q_{5})$ and the $q_j$'s are positive integers such that $\frac{\lambda_1-\lambda_j}{\lambda_1}=\frac{p_j}{q_j}$ for some integers $p_j$ with $\operatorname{gcd}(p_j,q_j)=1$.
\item If $r\geq 2$  and $\lambda_j\neq m$ for each $j\in\{1,\ldots,r\}$, then
\begin{equation}
\label{ldisL}
\rho_{X\vee Y}=\left(\frac{q_{r+2}q_{r+3}}{\operatorname{gcd}(q,q_{r+2})\operatorname{gcd}(R,q_{r+3})}\right)\rho_X,
\end{equation}
where $R=\operatorname{lcm}(q_3,\ldots,q_{r+2})$, the $q_j$'s are positive integers such that $\frac{\lambda_1-\lambda_j}{\lambda_1-\lambda_2}=\frac{p_{j}}{q_{j}}$ for some integers $p_j$ with $\operatorname{gcd}(p_{j},q_{j})=1$ and $\lambda_1>\lambda_2$. Further, the coefficient of $\rho_X$ in (\ref{ldisL}) is an integer.
\item If $r\geq 3$  and $\lambda_r= m$, then
\begin{equation}
\label{ldisL1}
\rho_{X\vee Y}=(Q/q)\rho_X,
\end{equation}
where $q=\operatorname{lcm}(q_3,\ldots,q_r,q_{r+1})$, $Q=\operatorname{lcm}(q_3,\ldots,q_{r-1},q_{r+1},q_{r+2},q_{r+3})$ and the $q_j$'s are positive integers satisfying $\frac{\lambda_1-\lambda_j}{\lambda_1}=\frac{p_j}{q_j}$ and $\operatorname{gcd}(p_j,q_j)=1$. In this case, $Q/q$ is an integer.
\end{enumerate}
\end{enumerate}
\end{theorem}

\begin{proof}
Let $X$ be connected. By Lemma \ref{supp}(1), $\sigma_u(L)=\{\lambda_1+n,\ldots,\lambda_r+n\}\cup\{\lambda_{r+2}+n,\lambda_{r+3}+n\}$. Note that $\lambda_j\neq \lambda_{r+3}$ for each $j\in\{1,\ldots,r\}$. However, it is possible that $\lambda_j=\lambda_{r+2}=m$ for some $j\in\{1,\ldots,r\}$. In this case, we may instead assume that $\lambda_r=m$. We have the following cases.
\begin{itemize}
\item Let $r=1$. Then $\sigma_u(L(X))=\{\lambda_1,0\}$ and $\sigma_u(L)=\{\lambda_1+n,m+n,0\}$, and so Theorem \ref{minperiod} yields $\rho_{X}=2\pi/\lambda_1$ and $\rho_{X\vee Y}=2\pi q/(\lambda_1+n)$, where $q=1$ whenever $\lambda_1=m$.
\item Let $r=2$ and $\lambda_r=m$. Then $\sigma_u(L)=\{\lambda_1+n,m+n,0\}$, and so $\rho_{X\vee Y}=2\pi q/(\lambda_1+n)$ as in the previous case. Since $\sigma_u(L(X))=\{\lambda_1,m,0\}$, Theorem \ref{minperiod} yields $\rho_{X}=2\pi q'/\lambda_1$.
\item Let $r\geq 2$ and $\lambda_j\neq m$ for each $j$. As $\sigma_u(L)=\{\lambda_1+n,\ldots,\lambda_r+n,\lambda_{r+2}+n,\lambda_{r+3}+n\}$ and $u$ is periodic in $X$ and $X\vee Y$, Theorem \ref{minperiod} yields $\rho_X=2\pi q/(\lambda_1-\lambda_2)$ and $\rho_{X\vee Y}=2\pi Q/(\lambda_1-\lambda_2)$, where $q=\operatorname{lcm}(q_3,\ldots,q_{r+1})$, $Q=\operatorname{lcm}(q_3,\ldots,q_r,q_{r+2},q_{r+3})$ and $\lambda_1>\lambda_2$. Since $\operatorname{lcm}(a,b,c)=\operatorname{lcm}(\operatorname{lcm}(a,b),c)$ and $\operatorname{lcm}(a,b)=\frac{ab}{\operatorname{gcd}(a,b)}$, we may write $q=\frac{R_1q_{r+1}}{\operatorname{gcd}(R_1,q_{r+1})}$ and
\begin{equation}
\label{Q}
Q=\frac{R_1q_{r+2}q_{r+3}}{\operatorname{gcd}(R_1,q_{r+2})\operatorname{gcd}(R_2,q_{r+3})}=qc,
\end{equation}
where $c$ is the coefficient of $\rho_X$ in (\ref{lconL}). This implies that $\rho_{X\vee Y}=c\rho_X$.
\item Let $r\geq 3$ and $\lambda_r=m$. Then $\sigma_u(L)=\{\lambda_1+n,\ldots,\lambda_{r-1}+n,\lambda_{r+2},\lambda_{r+3}\}$, and so the same argument in previous subcase yields (\ref{Q}) with $R_1=\operatorname{lcm}(q_3,\ldots,q_{r-1})$ and $R_2=\operatorname{lcm}(q_3,\ldots,q_{r-1},q_{r+2})$.
\end{itemize}
Combining the above cases proves (1).

Next, let $X$ be disconnected. By Lemma \ref{supp}(1), $\sigma_u(L)=\{\lambda_1+n,\ldots,\lambda_r+n\}\cup\{\lambda_{r+1}+n,\lambda_{r+2}+n,\lambda_{r+3}+n\}$. Moreover, the fact that $u$ is a non-isolated vertex yields $r\geq 1$. We have the following cases.

\begin{itemize}
\item Let $r=1$. Then $\sigma_u(L(X))=\{\lambda_1,0\}$ and $\sigma_u(L)=\{\lambda_1+n,\lambda_{2}+n,\lambda_{3}+n,\lambda_{4}+n\}$. By Theorem \ref{minperiod}, $\rho_{X}=2\pi/\lambda_1$ and $\rho_{X\vee Y}=2\pi q/\lambda_1$. In particular, if $\lambda_1=m$, then $q_1=q_3$, and so $q=q_4$.

\item Let $r=2$ and $\lambda_r=m$. Then $\sigma_u(L)=\{\lambda_1+n,\lambda_{3}+n,\lambda_{4}+n,\lambda_{5}+n\}$ and $\sigma_u(L(X))=\{\lambda_1,m,0\}$. By Theorem \ref{minperiod}, we get $\rho_{X\vee Y}=\frac{2\pi q}{\lambda_1-\lambda_3}=\frac{2\pi q}{\lambda_1}$ and $\rho_{X}=\frac{2\pi q'}{\lambda_1}$, where $p',q'\in\mathbb{Z}$ such that $\frac{m}{\lambda_1}=\frac{p'}{q'}\in\mathbb{Q}$ and $\operatorname{gcd}(p',q')=1$. Since $\frac{p_4}{q_4}=\frac{\lambda_1-\lambda_4}{\lambda_1-\lambda_3}=\frac{\lambda_1-m}{\lambda_1}=1-\frac{p'}{q'}$, we get $q_4=q'$.

\item Let $r\geq 2$ and $\lambda_j\neq m$ for each $j$. Then $\sigma_u(L)=\{\lambda_1+n,\ldots,\lambda_r+n,\lambda_{r+1}+n,\lambda_{r+2}+n,\lambda_{r+3}+n\}$. By Theorem \ref{minperiod}, $\rho_{X\vee Y}=2\pi Q/(\lambda_1-\lambda_2)$ where $Q=\operatorname{lcm}(q_3,\ldots,q_{r+3})$ and $q_j$'s are as defined in (1c). The same argument then yields $Q=\frac{qq_{r+2}q_{r+3}}{\operatorname{gcd}(q,q_{r+2})\operatorname{gcd}(R,q_{r+3})}=qc,$ where $q=\operatorname{lcm}(q_3,\ldots,q_{r+1})$.

\item Let $r\geq 3$ and $\lambda_r=m$. Since $\sigma_u(L)=\{\lambda_1+n,\ldots,\lambda_{r-1}+n,\lambda_{r+1}+n,\lambda_{r+2}+n,\lambda_{r+3}+n\}$, we get $\rho_X=2\pi q/\lambda_1$ and $\rho_{X\vee Y}=2\pi Q/\lambda_1$ by Theorem \ref{minperiod}. Now, since $\frac{\lambda_1-\lambda_r}{\lambda_1-\lambda_{r+1}}=\frac{\lambda_1-m}{\lambda_1}=\frac{p_r}{q_r}$ and $\frac{\lambda_1-\lambda_r}{\lambda_1-\lambda_{r+2}}=\frac{\lambda_1-m}{\lambda_1}=\frac{p_{r+2}}{q_{r+2}}$, we get $q_r=q_{r+2}$, and so $Q=\operatorname{lcm}(q_2,\ldots,q_{r+1},q_{r+3})$. Thus, $Q/q\in\mathbb{Z}$.
\end{itemize}

Combining the above cases proves (2).
\end{proof}

Using the same argument, we get an analogue of Theorem \ref{perjoin} for the adjacency case.

\begin{theorem}
\label{perjoinn}
Let $u$ be a non-isolated and periodic vertex in $X$ and $X\vee Y$ with $\sigma_u(A(X))=\{\lambda_1,\ldots,\lambda_r,\lambda_{r+1}\}$, where $\lambda_{r+1}=k$. Let $\lambda_{r+2}:=\lambda^-$ and $\lambda_{r+3}:=\lambda^+$, where $\lambda^{\pm}$ are defined in Lemma \ref{aqjoin}.

\begin{enumerate}
\item Suppose $X$ is connected.
\begin{enumerate}

\item Let $r=1$. Then $\rho_{X\vee Y}=\left(\frac{q(k-\lambda_1)}{\sqrt{D}}\right)\rho_X$, where $q$ is a positive integer such that $\frac{\lambda^+-\lambda_1}{\sqrt{D}}=\frac{p}{q}$ for some integer $p$ with $\operatorname{gcd}(p,q)=1$. In particular, if $\lambda_1=\lambda^-$, then $q=1$.

\item If $r=2$ and $\lambda_2= \lambda^-$, then $\rho_{X\vee Y}=\left(\frac{(\lambda_1-\lambda^-)q}{q'\sqrt{D}}\right)\rho_X$, where $q$ is an integer in (1b) and $q'$ is a positive integer such that $\frac{\lambda_1-k}{\lambda_1-\lambda^-}=\frac{p'}{q'}\in\mathbb{Q}$ for some integer $p'$ with $\operatorname{gcd}(p',q')=1$.

\item If $r\geq 2$ and $\lambda_j\neq  \lambda^-$ for each $j\in\{1,\ldots,r\}$, then the conclusion in Theorem \ref{perjoin}(1d) holds.

\item If $r\geq 3$ and $\lambda_r= \lambda^-$, then the conclusion in Theorem \ref{perjoin}(1e) holds.

\end{enumerate}

\item Suppose $X$ is disconnected. 
\begin{enumerate}
\item Let $r=1$. Then $\rho_{X\vee Y}=q\rho_X$, where $q=\operatorname{lcm}(q_{3},q_{4})$ and the $q_j$'s are positive integers such that $\frac{k-\lambda_j}{k-\lambda_1}=\frac{p_j}{q_j}$ for some $p_j$ with $\operatorname{gcd}(p_j,q_j)=1$. In particular, if $\lambda_1=\lambda^-$, then $q=q_4$.

\item If $r=2$ and $\lambda_2=\lambda^-$, then $\rho_{X\vee Y}=\left(q/q_2\right)\rho_X$, where $q=\operatorname{lcm}(q_{2},q_{5})$, the $q_j$'s are positive integers such that $\frac{k-\lambda_j}{k-\lambda_1}=\frac{p_j}{q_j}$ for some integers $p_j$ with $\operatorname{gcd}(p_j,q_j)=1$ and $\lambda_1>\lambda_2$.
\item If $r\geq 2$  and $\lambda_j\neq \lambda^-$ for each $j\in\{1,\ldots,r\}$, then the conclusion in Theorem \ref{perjoin}(2d) holds.
\item If $r\geq 3$  and $\lambda_r= \lambda^-$, then the conclusion in Theorem \ref{perjoin}(2e) holds.
\end{enumerate}
\end{enumerate}
\end{theorem}

\begin{remark}
\label{ratint}
The advantage of Theorems \ref{perjoin}-\ref{perjoinn} is that they apply even if $\phi(M(X),t)\notin\mathbb{Z}[t]$. We also mention that the constant $c>0$ such that $\rho_{X\vee Y}=c\rho_{X}$ in Theorems \ref{perjoin}-\ref{perjoinn} is always rational. In fact, $c$ is an integer whenever (i) $X$ is disconnected or (ii) $X$ is connected and $q_{r+1}=1$ in Theorems \ref{perjoin}-\ref{perjoinn}(1c-d).
\end{remark}

For the family of complete graphs, Theorem \ref{perjoin}(1a) yields $\rho_{X\vee Y}=\frac{m}{m+n}\rho_{X}$, where $X=K_m$ and $Y=K_n$. Here, $c=\frac{m}{m+n}<1$ is rational, and this holds for $M\in\{A,L\}$. We complement this observation by providing an infinite family of connected graphs $X$ such that $\rho_{X\vee Y}=c\rho_X$ for some integer $c\geq 1$. 

\begin{example}
Suppose that $m$ is even and let $X=K_{\frac{m}{2},\frac{m}{2}}$. Since $\sigma_u(L(X))=\{m/2,m,0\}$, each $u\in V(X)$ is periodic with $\rho_X=4\pi/m$. Consider $X\vee Y$, where $Y$ is a simple unweighted graph on $n$ vertices. Applying Theorem \ref{perjoin}(1b), we get $\rho_{X\vee Y}=c\rho_X$, where $c=m/2g$ and $g=\operatorname{gcd}(m/2,n)$. Since $g$ divides $m/2$, we conclude that $c$ is an integer.
\end{example}

\begin{example}
Let $X=K_m$ and $Y$ be a simple unweighted $\ell$-regular graph on $n$ vertices. Each $u\in V(X)$ is periodic with $\rho_X=2\pi/m$. By Lemma \ref{supp}(2), $\sigma_u(A)=\{\lambda^{\pm},-1\}$, where $\lambda^{\pm}=\frac{1}{2}(\ell+m-1\pm\sqrt{D})$. By Corollary \ref{perequiv}, $u$ periodic in $X\vee Y$ if and only if $D=(\ell-m+1)^2+4mn$ is a perfect square. Equivalently, $n=\frac{s(\ell-m+s+1)}{m}$ for some integer $s$ such that $m$ divides $s(\ell+s+1)$. Thus, if $u$ periodic in $X\vee Y$, then $\frac{\lambda^+-(-1)}{\lambda^+-\lambda^-}=\frac{\ell+s+1}{\ell-m+2s+1}\in\mathbb{Q}$, from which we get $\rho_{X\vee Y}=2\pi/g$, where $g=\operatorname{gcd}(\ell+s+1,s-m)$. Hence $\rho_{X\vee Y}=(m/g)\rho_X$. Taking $s=2m$ and $0\leq \ell\leq n-2$ yields $g=\operatorname{gcd}(\ell+1,m)$, and so $m/g$ is an integer.
\end{example}

\section{Strong cospectrality}
\label{secSc}

We say that two vertices $u$ and $v$ in a weighted graph $X$ are \textit{strongly cospectral} if $E_j\textbf{e}_u=\pm E_j\textbf{e}_v$ for each $j$, in which case we define the sets
\begin{equation}
\label{esupp+-}
\sigma_{uv}^+(M)=\{\lambda_j:E_j\textbf{e}_u=E_j\textbf{e}_v\neq \textbf{0}\}\quad \text{and}\quad \sigma_{uv}^-(M)=\{\lambda_j:E_j\textbf{e}_u=-E_j\textbf{e}_v\neq \textbf{0}\}.
\end{equation}
Equivalently, vertices $u$ and $v$ are strongly cospectral if for each $j$, either every eigenvector $\textbf{w}$ associated with $\lambda_j$ satisfies $\textbf{w}^T\textbf{e}_u=\textbf{w}^T\textbf{e}_v$ or every eigenvector $\textbf{w}$ associated with $\lambda_j$ satisfies $\textbf{w}^T\textbf{e}_u=-\textbf{w}^T\textbf{e}_v$.

If vertices $u$ and $v$ are strongly cospectral, then $\sigma_u(M)=\sigma_{uv}^+(M)\cup \sigma_{uv}^-(M)$, and $u$ and $v$ belong to the same connected component of $X$ (so neither of them is isolated). In order to avoid confusion, if $u$ and $v$ are strongly cospectral in $X$ and $X\vee Y$, then we write the sets in (\ref{esupp+-}) as $\sigma_{uv}^{\pm}(M(X))$ and $\sigma_{uv}^{\pm}(M)$, respectively.

Strong cospectrality is a necessary condition for PST. Hence, in order to investigate PST in joins, we first need to characterize strong cospectrality. In what follows, we let $O_n(k)$ denote the empty graph on $n$ vertices with a loop of weight $k$ on each vertex. If $k=0$, then we simply write $O_n(k)$ as $O_n$.

\begin{theorem}
\label{sc}
Let  $m\geq 2$ and $u,v\in V(X)$ with $u\neq v$.
\begin{enumerate}
\item Vertices $u$ and $v$ in $X$ are Laplacian strongly cospectral in $X\vee Y$ if and only if either:
\begin{enumerate}
\item Vertices $u$ and $v$ are Laplacian strongly cospectral in $X$ and $m\notin\sigma_{uv}^-(L(X))$. In this case,
\begin{equation*}
\sigma_{uv}^+(L)=\{\lambda+n:0\neq \lambda\in\sigma_{uv}^+(L(X))\}\cup\mathcal{R}\quad \text{and}\quad \sigma_{uv}^-(L)=\{\mu+n:\mu\in\sigma_{uv}^-(L(X))\}
\end{equation*}
where $\mathcal{R}=\{0,m+n\} $ if $X$ is connected and $\mathcal{R}=\{0,m+n,n\} $ otherwise.
\item $X=O_2$. In this case, $\sigma_{uv}^+(L)=\{0,n+2\}$ and $\sigma_{uv}^-(L)=\{n\}$.
\end{enumerate}
\item Vertices $u$ and $v$ in $X$ are adjacency strongly cospectral in $X\vee Y$ if and only if either:
\begin{enumerate}
\item Vertices $u$ and $v$ are adjacency strongly cospectral in $X$ and $\lambda^{\pm}\notin\sigma_{uv}^-(L(X))$. In this case,
\begin{equation*}
\sigma_{uv}^+(A)=\sigma_{uv}^+(A(X))\backslash\{k\}\cup\mathcal{R}\quad \text{and}\quad \sigma_{uv}^-(A)=\sigma_{uv}^-(A(X))
\end{equation*}
where $\mathcal{R}=\{\lambda^{\pm}\} $ if $X$ is connected and $\mathcal{R}=\{\lambda^{\pm},k\} $ otherwise.
\item $X=O_2(k)$. In this case, $\sigma_{uv}^+(L)=\{\lambda^\pm\}$ and $\sigma_{uv}^-(L)=\{k\}$.
\end{enumerate}
\item If $w\in V(Y)$, then vertices $u$ and $w$ are not strongly cospectral in $X\vee Y$.
\end{enumerate}
Moreover, if (1a) or (2a) holds, then vertices $u$ and $v$ belong to the same connected component in $X$.
\end{theorem}

\begin{proof}
Let $\lambda$ and $\mu$ be nonzero eigenvalues of $L(X)$ and $L(Y)$ with associated eigenvectors $\textbf{y}_{\lambda}$ and $\textbf{z}_{\mu}$, respectively. Then $0$, $m+n$, $\lambda+n$ and $\mu+n$ are eigenvalues of $L$ with associated eigenvectors
\begin{equation}
\label{Elamb}
\textbf{1}_{m+n},\quad \textbf{u}=\left[ \begin{array}{ccccc} n\textbf{1}_m \\ -m\textbf{1}_n\end{array} \right],\quad  \textbf{v}_{\lambda}=\left[ \begin{array}{cc} \textbf{y}_{\lambda} \\ \textbf{0}\end{array} \right],\quad \text{and}\quad \textbf{w}_{\mu}=\left[ \begin{array}{cc} \textbf{0} \\ \textbf{z}_{\mu}\end{array} \right]
\end{equation}
respectively. Now, let $u,v\in V(X)$. We have two cases.
\begin{itemize}
\item Let $u$ and $v$ be non-isolated in $X$. From the form of $\textbf{v}_\lambda$'s in (\ref{Elamb}), strong cospectrality in $X$ is required for strong cospectrality in $X\vee Y$. Hence, we assume $u$ and $v$ are strongly cospectral in $X$. If $m\in\sigma_{uv}^-(L(X))$, then $\textbf{v}_m$ is another eigenvector for $m+n$. In this case, $\textbf{u}^T\textbf{e}_u= \textbf{u}^T\textbf{e}_v$ and $\textbf{v}_m^T\textbf{e}_u=-\textbf{v}_m^T\textbf{e}_v$, and so $u$ and $v$ are not strongly cospectral in $X\vee Y$. However, if $m\notin\sigma_{uv}^-(L(X))$, then $u$ and $v$ are strongly cospectral in $X\vee Y$ with $\sigma_{uv}^+(L)$ and $\sigma_{uv}^-(L)$ in Theorem \ref{sc}(1a) as desired.

\item Let $u$ and $v$ be isolated in $X$. Then $\textbf{v}=\textbf{e}_u-\textbf{e}_v$ is an eigenvector for $L$ associated to the eigenvalue $n$. If $X=O_2$, then \cite[Corollary 6.9(2)]{Monterde2022} yields the desired conclusion. However, if $X$ has a connected component $C$ other than $\{u\}$ and $\{v\}$, then the vector $\textbf{w}$ which is constant on each connected component of $X$ and whose sum of all entries is 0 is also an eigenvector for $L$ associated to the eigenvalue $n$. This vector satisfies $\textbf{w}^T\textbf{e}_u=\textbf{w}^T\textbf{e}_v$, and since $\textbf{v}^T\textbf{e}_u=-\textbf{v}^T\textbf{e}_v$, we get that $u$ and $v$ are not strongly cospectral in $X\vee Y$.
\end{itemize}
Combining these cases proves (1). Next, let $\lambda\neq k$ and $\mu\neq \ell$ be eigenvalues of $A(X)$ and $L(Y)$ respectively, with associated eigenvectors $\textbf{y}_{\lambda}$ and $\textbf{z}_{\mu}$. Then $\lambda^{\pm}$, $\lambda$ and $\mu$ are eigenvalues of $A$ with associated eigenvectors
\begin{equation}
\label{Elamb1}
\textbf{u}=\left[ \begin{array}{ccccc} (k-\lambda^{\mp})\textbf{1}_m \\ m\textbf{1}_n\end{array} \right],\quad  \textbf{v}_{\lambda}=\left[ \begin{array}{cc} \textbf{y}_{\lambda} \\ \textbf{0}\end{array} \right],\quad \text{and}\quad \textbf{w}_{\mu}=\left[ \begin{array}{cc} \textbf{0} \\ \textbf{z}_{\mu}\end{array} \right]
\end{equation}
respectively. 
The same argument as the previous case yields the desired conclusion for (2). Finally, the form of the $\textbf{v}_{\lambda}$'s in (\ref{Elamb}-\ref{Elamb1}) yields (3) and the last statement.
\end{proof}

From Theorem \ref{sc}(3), we assume henceforth that strongly cospectral vertices in $X\vee Y$ belong to $X$.

Note that if $X$ is unweighted. then $m$ is an eigenvalue of $L(X)$ if and only if $X$ is a join. From (\ref{Elamb}), we get that $m\in\sigma_{uv}^-(L(Z))$ if and only if $Z=K_2$. Combining this with Theorem \ref{sc} yields the next result.

\begin{corollary}
\label{scj}
Let $X$ be an unweighted graph with vertices $u$ and $v$.
\begin{enumerate}
    \item Let $X\notin\{K_2,O_2\}$. Then $u$ and $v$ in $X$ are Laplacian strongly cospectral in $X\vee Y$ if and only if they are in $X$.
    \item Let $X\neq O_2$ and suppose $X\vee Y$ is not a complete graph. Then $u$ and $v$ in $X$ are adjacency strongly cospectral in $X\vee Y$ if and only if they are in $X$.
\end{enumerate}
\end{corollary}

\begin{proof}
Let $X$ be unweighted. Then $m$ is an eigenvalue of $L(Z)$ if and only if $Z$ is a join. In particular, $m\in\sigma_{uv}^-(L(Z))$ if and only if $Z=K_2$ by (\ref{Elamb}). Combining this with Theorem \ref{sc} yields (1). To prove (2), note that $\lambda^-=\frac{1}{2}(k+\ell-\sqrt{(k-\ell)^2+4mn})=-1$ if and only if $m=k+1$ and $n=\ell+1$, i.e., $X$ and $Y$ are complete.
\end{proof}

\begin{example}
\label{scjoin}
Consider $X\vee Y$, where $X\in\{K_m,O_m\}$. If $m\geq 3$, then the vertices of $X$ do not admit strongly cospectrality in $X$ and $X\vee Y$ with respect to $M\in\{A,L\}$ \cite[Corollary 3.10]{Monterde2022}. If $m=2$, then the following hold about the apexes $u$ and $v$ of the double cone $X\vee Y$.
\begin{enumerate}
\item By Theorem \ref{sc}(1b-2b), $u$ and $v$ are Laplacian and adjacency strongly cospectral in $O_2\vee Y$.
\item Since $2\in\sigma_{uv}^-(L(X))$, Theorem \ref{sc}(1a) implies that $u$ and $v$ are not Laplacian strongly cospectral in $K_2\vee Y$ for any graph $Y$. Since $u$ and $v$ are adjacency strongly cospectral in $K_2$, Corollary \ref{scj}(2) implies that $u$ and $v$ are strongly cospectral in $K_2\vee Y$ if and only if $Y$ is not a complete graph. 
\end{enumerate}
\end{example}

\section{Perfect state transfer}\label{secPst}

To determine PST in joins, we make use of a characterization of PST due to Coutinho. Throughout, we denote the largest power of two that divides an integer $a$ by $\nu_2(a)$. We also denote the minimum PST times between $u$ and $v$ in $X$ and $X\vee Y$ by $\tau_X$ and $\tau_{X\vee Y}$, respectively.

\begin{theorem}
\label{pstC}
Let $M$ be a real symmetric matrix such that $\phi(M,t)\in\mathbb{Z}[t]$. Then vertices $u$ and $v$ in $X$ admit perfect state transfer if and only if all of the following conditions hold.
\begin{enumerate}
\item Either (i) $\sigma_u(M)\subseteq\mathbb{Z}$ or (ii) each $\lambda_j\in \sigma_u(M)$ is a quadratic integer of the form $\frac{1}{2}(a+b_j\sqrt{\Delta})$ for some integers $a$, $b_j$ and $\Delta$, where $\Delta>1$ is square-free. Here, we let $\Delta=1$ whenever (i) holds.
\item Vertices $u$ and $v$ are strongly cospectral in $X$.
\item For all $\lambda,\eta\in\sigma_{uv}^+(M)$ and $\mu,\theta\in\sigma_{uv}^-(M)$, we have $\nu_2\left(\frac{\lambda-\eta}{\sqrt{\Delta}}\right)>\nu_2\left(\frac{\lambda-\mu}{\sqrt{\Delta}}\right)=\nu_2\left(\frac{\lambda-\theta}{\sqrt{\Delta}}\right)$.
\end{enumerate}
Further, $\tau_{X}=\frac{\pi}{g\sqrt{\Delta}}$, where $g=\operatorname{gcd}\left(\left\{\frac{\lambda_0-\lambda}{\sqrt{\Delta}}:\lambda\in\sigma_{u}(M)\right\}\right)$ for some fixed $\lambda_0\in\sigma_{uv}^+(M)$.
\end{theorem}

To characterize PST in joins, we now combine Theorem \ref{pstC} with our results on periodicity and strong cospectrality in joins from Sections \ref{secPer} and \ref{secSc}. Note that it suffices to consider the vertices of $X$ in checking whether PST occurs in $X\vee Y$ by virtue of Theorem \ref{sc}(3). We begin with the Laplacian case.

\subsection{Laplacian case}
\label{PGSTpstL}

\begin{theorem}
\label{pstL}
Let $m\geq 2$ and $\phi(L(X),t)\in\mathbb{Z}[t]$. Vertices $u$ and $v$ in $X$ admit Laplacian perfect state transfer in $X\vee Y$ if and only if all of the following conditions hold.
\begin{enumerate}
\item Either (i) $u$ and $v$ are Laplacian strongly cospectral in $X$ and $m\notin\sigma_{uv}^-(L(X))$ or (ii) $X=O_2$.
\item The eigenvalues in $\sigma_u(L(X))$ are all integers.
\item One of the following conditions hold.
\begin{enumerate}
\item $X$ is connected and one of the following conditions hold for all $0\neq \lambda\in\sigma_{uv}^+(L(X))\cup\{m\}$ and $\mu,\theta\in\sigma_{uv}^-(L(X))$.
\begin{enumerate}
\item $\nu_2(\lambda)>\nu_2(\mu)=\nu_2(\theta)$ and $\nu_2(n)>\nu_2(\mu)$.
\item $\nu_2(\mu)>\nu_2(\lambda)=\nu_2(n)$.
\item $\nu_2(\lambda)=\nu_2(\mu)=\nu_2(n)$ and $\nu_2\left(\frac{\lambda+n}{2^{\nu_2(n)}}\right)>\nu_2\left(\frac{\mu+n}{2^{\nu_2(n)}}\right)=\nu_2\left(\frac{\theta+n}{2^{\nu_2(n)}}\right)$.
\end{enumerate}
\item $X\neq O_2$ is disconnected and (ai) above holds. Here, $u$ and $v$ are in the same connected component in $X$.
\item $X=O_2$ and $n\equiv 2$ (mod 4).
\end{enumerate}
\end{enumerate}
Further, $\tau_{X\vee Y}=\frac{\pi}{g}$, where $g=\operatorname{gcd}\left(\mathcal{T}\right)$, $\mathcal{T}=\mathcal{S}\cup\mathcal{R}$ and the sets $\mathcal{S},\mathcal{R}$ are given in Lemma \ref{supp}. 
\end{theorem}

\begin{proof}
By Theorem \ref{sc}(1) and Corollary \ref{joinperpreserve}, conditions (1) and (2) are equivalent resp.\ to strong cospectrality and periodicity of $u$ and $v$ in $X\vee Y$. Combining this with Theorem \ref{pstC}, we get PST between $u$ and $v$ in $X\vee Y$ if and only if  Theorem \ref{pstC}(3) holds. To establish (3a), suppose $X$ is connected. From Theorem \ref{sc}(1a), we have $\sigma_{uv}^+(L)= \{\lambda+n:0\neq \lambda\in\sigma_{uv}^+(L(X))\}\cup \{0,m+n\}$ and $\sigma_{uv}^-(L)=\sigma_{uv}^-(L(X))$. Theorem \ref{pstC}(3) applied to $u$ and $v$ is equivalent to
\begin{equation}
\label{nu2}
\nu_2(\lambda+n)>\nu_2(\mu+n)=\nu_2(\theta+n)
\end{equation}
for all $0\neq \lambda\in\sigma_{uv}^+(L(X))\cup\{m\}$ and $\mu,\theta\in\sigma_{uv}^-(L(X))$. We have three cases.
\begin{itemize}
\item Let $\nu_2(\lambda)>\nu_2(\mu)$ for some $0\neq \lambda\in\sigma_{uv}^+(L(X))\cup\{m\}$ and $\mu\in\sigma_{uv}^-(L(X))$. Since $\nu_2(\lambda+n)>\nu_2(\mu+n)$ in (\ref{nu2}), it must be that $ \nu_2(n)>\nu_2(\mu)$ for each $\mu\in\sigma_{uv}^-(L(X))$, in which case, $\nu_2(\mu+n)=\nu_2(\mu)$. Now, for the equality in (\ref{nu2}) to hold, we need $\nu_2(\mu)=\nu_2(\theta)$ for all $\mu,\theta\in\sigma_{uv}^-(L(X))$. Hence, $\nu_2(\lambda)>\nu_2(\mu)$ for all $\mu\in\sigma_{uv}^-(L(X))$. Finally, if $\nu_2(\mu)\geq \nu_2(\eta)$ for some $0\neq \eta\in\sigma_{uv}^+(L(X))\cup\{m\}$, then $\nu_2(n)> \nu_2(\eta)$, and so $\nu_2(\eta+n)=\nu_2(\eta)\leq \nu_2(\mu)=\nu_2(\mu+n)$, a contradiction to (\ref{nu2}). Thus, our assumption in this case combined with (\ref{nu2}) gives us $\nu_2(\lambda)>\nu_2(\mu)=\nu_2(\theta)$ and $ \nu_2(n)>\nu_2(\mu)$ for all $0\neq \lambda\in\sigma_{uv}^+(L(X))\cup\{m\}$ and $\mu,\theta\in\sigma_{uv}^-(L(X))$, which proves (3ai).
\item Let $\nu_2(\mu)>\nu_2(\lambda)$ for some $0\neq \lambda\in\sigma_{uv}^+(L(X))\cup\{m\}$ and $\mu\in\sigma_{uv}^-(L(X))$. If $\nu_2(\lambda)\neq \nu_2(n)$, then $\nu_2(\mu+n)\geq \nu_2(\lambda+n)$, a contradiction to (\ref{nu2}). Hence, we must have $\nu_2(\lambda)=\nu_2(n)$ for each $0\neq \lambda\in\sigma_{uv}^+(L(X))\cup\{m\}$. Finally, if $\nu_2(n)\geq \nu_2(\theta)$ for some $\theta\in\sigma_{uv}^-(L(X))$, then we have $\nu_2(\theta+n)\neq \nu_2(n)=\nu_2(\mu+n)$, which again contradicts (\ref{nu2}). Thus, $\nu_2(\mu)>\nu_2(n)=\nu_2(\lambda)$ for all for some $0\neq \lambda\in\sigma_{uv}^+(L(X))\cup\{m\}$ and $\theta\in\sigma_{uv}^-(L(X))$. This proves (3aii).
\item Let $\nu_2(\mu)=\nu_2(\lambda)$ for some $0\neq \lambda\in\sigma_{uv}^+(L(X))\cup\{m\}$ and $\mu\in\sigma_{uv}^-(L(X))$. For the strict inequality in (\ref{nu2}) to hold, it is required that $\nu_2(\mu)=\nu_2(\lambda)=\nu_2(n)$ for all $0\neq \lambda\in\sigma_{uv}^+(L(X))\cup\{m\}$ and $\mu\in\sigma_{uv}^-(L(X))$. From this, the conditions in (3aiii) are straightforward.
\end{itemize}
Combining these cases proves (3a). Now, if $X$ is disconnected, then $n\in\sigma_{uv}^+(L(X))$ by Theorem \ref{sc}(1b), and so Theorem \ref{pstC}(3) holds if and only if (\ref{nu2}) holds and $\nu_2(n)>\nu_2(\mu+n)$ for all $\mu\in\sigma_{uv}^-(L(X))$. Equivalently, (3ai) holds. This proves (3b). If $X=O_2$, then Theorem \ref{pstC}(3) holds if and only if $\nu_2(n+2)>\nu_2(n)$, i.e., $n\equiv 2$ (mod 4). This proves (3c). The minimum PST time follows from Theorem \ref{pstC}. 
\end{proof}

\begin{remark}
If $X\neq K_2$ is unweighted, then we may drop the condition $m\in\sigma_{uv}^-(L(X))$ in Theorem \ref{pstL}(1i).
\end{remark}

The following result is immediate from Theorem \ref{pstL}(3). 

\begin{corollary}
\label{pstLc3}
Let  $m\geq 2$ and $\phi(L(X),t)\in\mathbb{Z}[t]$. If $m+n$ is odd, then $X\vee Y$ has no Laplacian perfect state transfer. Moreover, if $m$ or $n$ is odd and $X$ is either disconnected or admits Laplacian perfect state transfer, then $X$ does not admit Laplacian perfect state transfer in $X\vee Y$.
\end{corollary}

\subsection{Adjacency case}
\label{PGSTpstA}

\begin{theorem}
\label{pstA}
Let $m\geq 2$, $k,\ell\in\mathbb{Z}$ and $\phi(A(X),t)\in\mathbb{Z}[t]$. Vertices $u$ and $v$ in $X$ admit adjacency perfect state transfer in $X\vee Y$ if and only if all of the following conditions hold.
\begin{enumerate}
\item Either (i) $u$ and $v$ are adjacency strongly cospectral in $X$ and $\lambda^-\notin\sigma_{uv}^-(L(X))$ or (ii) $X=O_2(k)$.
\item One of the following conditions hold.
\begin{enumerate}
\item The eigenvalues in $\sigma_u(A(X))$ are all integers and $D$ is a perfect square.
\item $X$ is connected, each $\lambda_j\in\sigma_u(A(X))\backslash\{k\}$ is of the form $\frac{1}{2}(k+\ell+b_j\sqrt{\Delta})$ and $\lambda^{\pm}=\frac{1}{2}(k+\ell\pm b\sqrt{\Delta})$, where $b_j,b,\Delta$ are integers with $b>b_j$ for each $j$ and $\Delta>1$ is square-free. 
\end{enumerate}
\item For all $\lambda,\eta\in \sigma_{uv}^+(A(X))\backslash\{k\}\cup \mathcal{R}$ and $\mu,\theta\in\sigma_{uv}^-(A(X))$, $\nu_2\left(\frac{\lambda-\eta}{\sqrt{\Delta}}\right)>\nu_2\left(\frac{\lambda-\mu}{\sqrt{\Delta}}\right)=\nu_2\left(\frac{\lambda-\theta}{\sqrt{\Delta}}\right)$, where $\Delta=1$ whenever (2a) holds, $\mathcal{R}$ is given in Theorem \ref{sc}(2).
\end{enumerate}
Further, $\tau_{X\vee Y}= \frac{\pi}{g\sqrt{\Delta}}$, where $g= \operatorname{gcd}\left(\mathcal{T}\right)$, $\mathcal{T}=\left\{\frac{\lambda_0-\lambda}{\sqrt{\Delta}}: \lambda\in\sigma_{u}(A)\backslash\{k\} \cup \mathcal{R}\right\}$ and $\lambda_0\in\sigma_{uv}^+(A)\backslash\{k\} \cup \mathcal{R}$ is fixed.
\end{theorem}

\begin{proof}
This is a direct consequence of Theorems \ref{joinper1}(2), \ref{sc}(2) and \ref{pstC}.
\end{proof}

The next result can be viewed as an analogue of Corollary \ref{pstLc3}. 

\begin{corollary}
\label{pstAc3}
If $k+\ell$ is odd in Theorem \ref{pstA}, then $X\vee Y$ has no adjacency perfect state transfer.
\end{corollary}
\begin{proof}
If $k+\ell=\lambda^++\lambda^+$ is odd, then so is $\lambda^+-\lambda^-$, a contradiction to Theorem \ref{pstA}(3).
\end{proof}

\section{PST in $X$ and $X\vee Y$}\label{secPst1}

We now utilize our results from the previous section to determine when PST is preserved in the join.

\subsection{Laplacian case}

\begin{corollary}
\label{pstLc2}
Let $m\geq 2$ and $\phi(L(X),t)\in\mathbb{Z}[t]$. Suppose Laplacian perfect state transfer occurs between vertices $u$ and $v$ in $X$ and $\tau_X=\frac{\pi}{h}$. Then it occurs between $u$ and $v$ in $X\vee Y$ if and only if $m\notin\sigma_{uv}^-(L(X))$ and $m,n\equiv 0$ (mod $2^{\alpha}$), where $\alpha>\nu_2(h)$. In this case, $\tau_{X\vee Y}=\frac{\pi}{g}$, where $g=\operatorname{gcd}(\mathcal{T})$, $\mathcal{T}$ is given in Theorem \ref{pstL} and $\nu_2(g)=\nu_2(h)$. Moreover, if $X$ is connected, then $\frac{h}{g}$ is rational. Otherwise, $\frac{h}{g}=\operatorname{lcm}(h/h_1,h/h_2)$ is an odd integer where $h_1=\operatorname{gcd}(m,h)$ and $h_2=\operatorname{gcd}(n,h)$.
\end{corollary}

\begin{proof}
If PST occurs in $X$, then $u$ and $v$ are strongly cospectral in $X$ and Theorem \ref{pstL}(2) holds. Hence, PST occurs in $X\vee Y$ if and only if (1i) and (3ai) of Theorem \ref{pstL} holds. Equivalently, $m\notin\sigma_{uv}^-(L(X))$, $\nu_2(m)>\nu_2(\mu)$ and $\nu_2(n)>\nu_2(\mu)$ for all $\mu\in\sigma_{uv}^-(L(X))$. As $\nu_2(h)=\nu_2(\mu)=\nu_2(\mu+n)$ for all $\mu\in\sigma_{uv}^-(L(X))$, these conditions are equivalent to $m,n\equiv 0$ (mod $2^{\alpha}$), where $\alpha>\nu_2(h)$. The minimum PST time follows from Theorem \ref{pstL}.
\end{proof}

\begin{corollary}
\label{pstLc6}
Let $X$ be an unweighted graph on $m=2^p$ vertices. Then conditions (3aii-3aiii) of Theorem \ref{pstL} do not hold. If we add that Laplacian perfect state transfer occurs between vertices $u$ and $v$ in $X$, then:
\begin{enumerate}
\item If $p=1$, then Laplacian perfect state transfer does not occur between $u$ and $v$ in $K_2\vee Y$ for any $Y$.
\item If $p\geq 2$, then Laplacian perfect state transfer occurs between $u$ and $v$ in $X\vee Y$ if and only if $n\equiv 0$ (mod $2^{\alpha}$), where $\alpha>\nu_2(g)$ and $\frac{\pi}{h}$ is the minimum PST time between $u$ and $v$ in $X$.
\end{enumerate}
\end{corollary}

\begin{proof}
The assumption implies that each eigenvalue $\lambda$ of $L(X)$ is at most $m$. Thus, $\nu_2(\lambda)<\nu_2(m)$, and so conditions (3aii-iii) of  Theorem \ref{pstC} do not hold. Now, from Corollary \ref{pstLc2}, PST occurs between $u$ and $v$ in $X\vee Y$ if and only if $m\notin\sigma_{uv}^-(L(X))$ and $n\equiv 0$ (mod $2^{\alpha}$), where $\alpha>\nu_2(h)$. Thus, if $p=1$, then Example \ref{scjoin}(1) implies that $u$ and $v$ are not strongly cospectral in $K_2\vee Y$, and so (1) holds. However, if $p\geq 2$, then $m\notin\sigma_{uv}^-(L(X))$, and so (2) holds.
\end{proof}

In Corollaries \ref{pstLc2} and \ref{pstLc6}(2), the $Y$ with the least number of vertices and edges that works is $Y=O_{n}$, where $n=2^{\nu_2(g)+1}$. Next, combining Corollary \ref{pstLc2} and Theorem \ref{pstL}(3ai) also yields the following result.

\begin{corollary}
\label{pstLc4}
Let $u$ and $v$ be strongly cospectral vertices in $X\vee Y$ and $\sigma_u(L(X))\subseteq\mathbb{Z}$. If condition (3ai)  of Theorem \ref{pstL} holds, then Laplacian perfect state transfer occurs between $u$ and $v$ in $X$ and $X\vee Y$.
\end{corollary}

From Theorem \ref{pstL}(3a), it is evident that PST in $X\vee Y$ does not necessarily yield PST in $X$. Thus, the above result can be viewed as a characterization of the equivalence of PST in a graph and its join.

Using Corollary \ref{pstLc2}, we provide an infinite family of graphs such that $X$ and $X\vee Y$ both have PST. Recall that a graph $X$ is \textit{Hadamard diagonalizable} if $L(X)$ is diagonalizable by a Hadamard matrix \cite{Barik2011}.

\begin{example}
\label{had}
Let $X$ be a Hadamard diagonalizable graph on $m\geq 4$ vertices with PST between $u$ and $v$. From \cite{Johnston2017}, $\sigma_{uv}^+(L(X))$ and $\sigma_{uv}^-(L(X))$ respectively consist of integers $\lambda\equiv 0$ and $\mu\equiv 2$ (mod 4). Since $m\equiv 0$ (mod 4), Corollary \ref{pstLc2} yields PST between $u$ and $v$ in $X\vee Y$ if and only if $n\equiv 0$ (mod 4). 
\end{example}

Our next result shows that for a graph $X$ with PST between $u$ and $v$, it is possible for PST to fail between $u$ and $v$ in $X\vee Y$ for any graph $Y$, yet occur in $(X\cup Z)\vee Y$ for some graph $Z$.

\begin{corollary}
\label{pstLdisc}
Let $m\geq 2$ and $\phi(L(X),t)\in\mathbb{Z}[t]$. If Laplacian perfect state transfer occurs between vertices $u$ and $v$ in $X$ with $m\in\sigma_{uv}^-(L(X))$, then the following hold.
\begin{enumerate}
\item For any graph $Y$, Laplacian perfect state transfer does not occur between $u$ and $v$ in $X\vee Y$.
\item Let $Z$ be a graph on $r$ vertices such that $m+r\notin\sigma_{uv}^-(L(X))$. Then Laplacian perfect state transfer occurs between $u$ and $v$ in $(X\cup Z)\vee Y$ if and only if $\nu_2(n)>\nu_2(m)=\nu_2(r)$.
\end{enumerate}
\end{corollary}

\begin{proof}
Since $m\in\sigma_{uv}^-(L(X))$, $u$ and $v$ are not strongly cospectral in $X\vee Y$ by Theorem \ref{sc}(1). We now prove (2). Since $m+r\notin\sigma_{uv}^-(L(X))$, Corollary \ref{pstLc2} implies that PST occurs between $u$ and $v$ in $X\vee Y$ if and only if $\nu_2(m+r)>\nu_2(m)$ and $\nu_2(n)>\nu_2(m)$ for all $\mu\in\sigma_{uv}^-(L(X))$. Equivalently, $\nu_2(n)>\nu_2(m)=\nu_2(r)$.
\end{proof}

In Corollary \ref{pstLdisc}(2), the $Z$ and $Y$ with the least number of vertices and edges that work are $Z=O_r$ and $Y=O_n$ respectively, where $r=2^{\nu_2(m)}$ and $r=2^{\nu_2(m)+1}$. 

\begin{example}
Let $X=K_2$ with vertices $u$ and $v$. From Corollary \ref{pstLc6}(1), PST between $u$ and $v$ fails in $K_2\vee Y$ for any graph $Y$. Since $r+2\notin\sigma_{uv}^-(L(K_2))$ for all $r\geq 1$, Corollary \ref{pstLdisc} yields PST between $u$ and $v$ in $(K_2\cup Z)\vee Y$ if and only if $\nu_2(n)>\nu_2(r)=1$. In particular, we may take $Z=O_2$ and $Y=O_4$.
\end{example}

The following remark can be used to generate infinite families of weighted graphs $X$ whereby PST occurs between $u$ and $v$ in $X$ and $(X\cup Z)\vee Y$ for some graph $Z$, but not in $X\vee Y$ for any graph $Y$.

\begin{remark}
\label{remLpst}
Let $m\geq 2$ and $\phi(L(X),t)\in\mathbb{Z}[t]$. Suppose vertices $u$ and $v$ admit Laplacian PST in $X$, and suppose $\mu$ divides $m$ for some $\mu\in\sigma_{uv}^-(L(X))$. Let $X'$ be the graph obtained by scaling all edge weights of $X$ by a factor of $\frac{m}{\mu}$. Then $m\in\sigma_{uv}^-(L(X'))$, and so PST does not occur between $u$ and $v$ in $X'\vee Y$ for any graph $Y$ by Corollary \ref{pstLdisc}(1). But since $\frac{m}{\mu}$ is an integer, we get $\phi(L(X'),t)\in\mathbb{Z}[t]$, and so we may apply Corollary \ref{pstLdisc}(2) to obtain a graph $Z$ such that PST occurs between $u$ and $v$ in $(X'\cup Z)\vee Y$.
\end{remark}

\subsection{Adjacency case}

In contrast with the Laplacian case in Corollary \ref{pstLdisc}(2), it turns out that if $X$ is regular and has adjacency PST between two vertices, then we can always choose a regular graph $Y$ such that adjacency PST is preserved in $X\vee Y$.

\begin{corollary}
\label{adjpstinX}
Let $m\geq 2$, $k,\ell\in\mathbb{Z}$ and $\phi(A(X),t)\in\mathbb{Z}[t]$. Suppose adjacency perfect state transfer occurs between vertices $u$ and $v$ in $X$ and $\tau_X=\frac{\pi}{h}$. Then it occurs between them in $X\vee Y$ if and only if 
\begin{enumerate}
\item $n=\frac{s(k-\ell+s)}{m}$ for some integer $s$ such that $m$ divides $s(k-\ell+s)$ and $\ell-s\notin\sigma_{uv}^{-}(A(X))$, and
\item for all $\mu\in \sigma_{uv}^{-}(A(X))$, $\nu_2(s)>\nu_2(k-\mu)$ and $\nu_2(\ell-k)>\nu_2(k-\mu)$.
\end{enumerate}
In this case, $\tau_{X\vee Y}=\frac{\pi}{g}$, where $g=\operatorname{gcd}(\mathcal{T})$, $\mathcal{T}$ is given in Theorem \ref{pstA}, $\lambda_0\in\sigma_{uv}^+(A)\backslash\{k\} \cup \mathcal{R}$ is fixed and $\nu_2(g)=\nu_2(h)$. Moreover, if $X$ is connected, $\frac{h}{g}$ is rational. Otherwise, $\frac{h}{g}=\operatorname{lcm}(h/h_1,h/h_2)$ is an odd integer and $h_1=\operatorname{gcd}(\lambda_0-\lambda^+,h)$ and $h_2=\operatorname{gcd}(\lambda_0-\lambda^-,h)$.
\end{corollary}

\begin{proof}
From Theorem \ref{pstA}, it follows that PST occurs between $u$ and $v$ in $X\vee Y$ if and only if
\begin{itemize}
\item[(a)] $D$ is a perfect square and $\lambda^-\notin\sigma_{uv}^-(A(X))$, and
\item[(b)]  $\nu_2(\lambda^{\pm}-\eta)>\nu_2(\lambda^{\pm}-\mu)=\nu_2(\lambda^{\pm}-\theta)$ for all $\eta\in \sigma_{uv}^+(A(X))\backslash\{k\}\cup \mathcal{S}$ and $\mu,\theta\in\sigma_{uv}^-(A(X))$.
\end{itemize}
By the choice of $s$ in (1), $D=(k-\ell)^2+4mn=(k-\ell+2s)^2$. Since $\lambda^{\pm}=\frac{1}{2}(k+\ell\pm\sqrt{D})$, we get $\lambda^{+}=k+s$ and $\lambda^{-}=\ell-s$. Thus, (1) ensures that condition (a) above holds. Using the fact that $X$ has PST between $u$ and $v$, Theorem \ref{pstC}(3) yields
\begin{equation}
\label{pstX}
\nu_2(k-\eta)>\nu_2(k-\mu)=\nu_2(k-\theta)
\end{equation}
for all $\eta\in \sigma_{uv}^+(A(X))\backslash\{k\}$ and $\mu,\theta\in\sigma_{uv}^-(A(X))$. Observe that for each $\lambda\in\sigma_{u}(A(X))\backslash\{k\}$, we can write $\lambda^+-\lambda=(\lambda^+-k)+(k-\lambda)=s+(k-\lambda)$. Similarly, $\lambda^--\lambda=(\ell-s-k)+(k-\lambda)$. Thus,
\begin{equation}
\label{pstX1}
\nu_2(\lambda^+-\lambda)\geq \min\{\nu_2(s),\nu_2(k-\lambda)\}\quad \text{and}\quad \nu_2(\lambda^--\lambda)\geq \min\{\nu_2(\ell-s-k),\nu_2(k-\lambda)\}.
\end{equation}
We have the following cases.
\begin{itemize}
\item Suppose $\nu_2(s)>\nu_2(k-\mu)$ for some $\mu\in\sigma_{uv}^-(A(X))$. Then (\ref{pstX}) implies that $\nu_2(s)>\nu_2(k-\mu)$ for all $\mu\in\sigma_{uv}^-(A(X))$, and combining this with (\ref{pstX1}) yields $\nu_2(\lambda^+-\mu)=\nu_2(\lambda^+-\theta)=\nu_2(k-\mu)$. Moreover, (\ref{pstX}) and (\ref{pstX1}) give us $\nu_2(\lambda^+-\eta)>\nu_2(k-\mu)$ for all $\mu\in\sigma_{uv}^-(A(X))$, from which we get $\nu_2(\lambda^+-\eta)>\nu_2(\lambda^+-\mu)=\nu_2(\lambda^+-\theta)$ for all $\eta\in \sigma_{uv}^+(A(X))\backslash\{k\}$ and $\mu,\theta\in\sigma_{uv}^-(A(X))$.
\item Suppose $\nu_2(s)\leq \nu_2(k-\mu)$ for some $\mu\in\sigma_{uv}^-(A(X))$. Then one can check using (\ref{pstX}) and (\ref{pstX1}) that $\nu_2(\lambda^+-\eta)=\nu_2(s)\leq \nu_2(\lambda^+-\mu)$, a violation of condition (b).
\end{itemize}
Consequently,  $\nu_2(\lambda^+-\eta)>\nu_2(\lambda^+-\mu)=\nu_2(\lambda^+-\theta)$ for all $\eta\in \sigma_{uv}^+(A(X))\backslash\{k\}$ and $\mu,\theta\in\sigma_{uv}^-(A(X))$ if and only if $\nu_2(s)>\nu_2(k-\mu)$ for all $\mu\in\sigma_{uv}^-(A(X))$. Using this fact and again arguing similarly as in the above two cases using (\ref{pstX}) and (\ref{pstX1}) yields $\nu_2(\lambda^+-\eta)>\nu_2(\lambda^+-\mu)=\nu_2(\lambda^+-\theta)$ for all $\eta\in \sigma_{uv}^+(A(X))\backslash\{k\}$ and $\mu,\theta\in\sigma_{uv}^-(A(X))$ if and only if $\nu_2(\ell-k)>\nu_2(k-\mu)$ for all $\mu\in\sigma_{uv}^-(A(X))$. Thus, the assumption in (2) is equivalent to condition (b) above, which yields the desired conclusion.
\end{proof}

The following example illustrates Corollary \ref{adjpstinX}.

\begin{example}
\label{hyp1}
Let $X$ be a simple integer-weighted Hadamard diagonalizable graph on $m\geq 4$ vertices with PST between vertices $u$ and $v$. Then $m\equiv 0$ (mod 4) and $X$ is $p$-regular for some integer $p$. From Example \ref{had}, $\sigma_{uv}^+(A(X))$ consists of integers $p-\lambda$, where $\lambda\equiv 0$ (mod 4), while $\sigma_{uv}^-(A(X))$ consists of integers  $p-\mu$ such that $\mu\equiv 2$ (mod 4). Invoking Corollary \ref{adjpstinX}(2), we get PST between $u$ and $v$ in $X\vee Y$ if and only if (i) $n=\frac{s(p-\ell+s)}{m}$, where $s$ is an integer such that $m$ divides $s(p-\ell+s)$ and $\ell-s\notin\sigma_{uv}^{-}(A(X))$ and (ii) $\nu_2(s)\geq 2$ and $\nu_2(\ell-p)\geq 2$. In particular, if $X=Q_p$, then we may take $\ell=p$ and $s=2^{(p+1)/2}a$ (i.e., $n=2a$) if $p$ is odd and $s=2^{p/2}a$ (i.e., $n=a$) otherwise, to get PST in both $Q_p$ and $Q_p\vee Y$.
\end{example}

\begin{remark}
If $X=K_2$, then $k=-\mu=1$, and so the condition $\ell-s\notin\sigma_{uv}^{-}(A(X))$ in Corollary \ref{adjpstinX}(1) is equivalent to $\ell\neq n-1$, while the conditions $\nu_2(s)>1$ and $\nu_2(\ell-1)>1$ in Corollary \ref{adjpstinX}(2) are equivalent to $\nu_2(s-2)=1$ and $\nu_2(\ell+3)>1$. Thus, taking $X=K_2$ in Corollary \ref{adjpstinX} recovers a known characterization of PST in connected double cones \cite[Theorem 12(1)]{Kirkland2023}.
\end{remark}

Next, we have an analogue of Corollary \ref{pstLc4}, which determines when PST in $X\vee Y$ is inherited by $X$.

\begin{corollary}
\label{pstAc1}
Let $m\geq 2$, $k,\ell\in\mathbb{Z}$ and $\phi(A(X),t)\in\mathbb{Z}[t]$.
\begin{enumerate}
\item Let $X$ be disconnected. If perfect state transfer occurs between vertices $u$ and $v$ in $X\vee Y$, then it also occurs between them in $X$ if and only if $X\neq O_2(k)$ for any $k\in\mathbb{Z}$.
\item Suppose $X$ is connected. If perfect state transfer occurs between vertices $u$ and $v$ in $X\vee Y$, then perfect state transfer occurs between $u$ and $v$ in $X$ if and only if both conditions below hold.
\begin{enumerate}
\item The eigenvalues in $\sigma_u(A(X))$ are all integers and $D$ is a perfect square.
\item $\nu_2(k-\eta)>\nu_2(k-\mu)=\nu_2(k-\theta)$ for any $ \eta\in \sigma_{uv}^+(A(X))\backslash\{k\}$ and $\mu,\theta\in\sigma_{uv}^-(A(X))$.
\end{enumerate}
\end{enumerate}
\end{corollary}

\section{PST in $X\vee Y$ but not in $X$}\label{secPst2}

In Corollary \ref{pstLc2}, we characterized graphs with PST that also admit PST in the join. In this section, our goal is to show that PST can be induced in the join despite the lack of PST in the underlying graph.

\subsection{Laplacian case}

The following is immediate from Theorem \ref{pstL}.

\begin{corollary}
\label{pstLc5}
Let $m\geq 2$ and $\phi(L(X),t)\in\mathbb{Z}[t]$. Suppose $u$ and $v$ are strongly cospectral vertices in $X$ with $\sigma_u(L(X))\subseteq\mathbb{Z}$. Then Laplacian perfect state transfer occurs between $u$ and $v$ in $X\vee Y$ but not in $X$ if and only if $m\notin\sigma_{uv}^-(L(X))$, $\nu_2(m)=\nu_2(n)$ and either (i) Theorem \ref{pstL}(3aii-3aiii) holds or (ii) $X=O_2$.
\end{corollary}

For strongly cospectral vertices in $X$ whose elements in $\sigma_{uv}^{-}(L(X))$ consists of integers whose largest powers of two in their factorizations are larger than those in $\sigma_{uv}^{+}(L(X))$, Theorem \ref{pstC} yields no PST between them. Nonetheless, we show that Laplacian PST can be induced in $X\vee Y$ by appropriate choice of $n$. The following result follows directly from Corollary \ref{pstLc5} and Theorem \ref{pstL}(3aii).

\begin{theorem}
\label{oddst1}
Let $m\geq 2$ and $\phi(L(X),t)\in\mathbb{Z}[t]$. Suppose $u$ and $v$ are strongly cospectral vertices in $X$ with $\sigma_u(L)\subseteq\mathbb{Z}$ such that the $\nu_2(\mu)>\nu_2(\lambda)$ for all $0\neq \lambda\in\sigma_{uv}^{+}(L(X))$ and $\mu\in\sigma_{uv}^{-}(L(X))$. Then Laplacian perfect state transfer occurs between $u$ and $v$ in $X\vee Y$ if and only if $m\notin\sigma_{uv}^-(L(X))$, $X$ is connected and $\nu_2(\lambda)=\nu_2(m)=\nu_2(n)$ for all $0\neq \lambda\in\sigma_{uv}^{+}(L(X))$.
\end{theorem}

\begin{example}
\label{cpm}
Let $X=CP(m)$, where $m\equiv 2$ (mod 4). Then $X$ is connected, $(m-2)$-regular and any two non-adjacent vertices $u$ and $v$ in $X$ are strongly cospectral with $\sigma_{uv}^{+}(L(X))=\{0,m\}$ and $\sigma_{uv}^{-}(L(X))=\{m-2\}$. As $\nu_2(m-2)>\nu_2(m)$, Laplacian PST does not occur between $u$ and $v$ in $X$ by Theorem \ref{pstC}(3). Applying Theorem \ref{oddst1}, Laplacian PST occurs between $u$ and $v$ in $X\vee Y$ if and only if $n\equiv 2$ (mod 4). In particular, if we take $n=2$, then $X\vee Y=CP(m+2)$ admits PST between any two non-adjacent vertices.
\end{example}

For strongly cospectral vertices in $X$ whose eigenvalue supports consist of integers with equal largest powers of two in their factorizations, Theorem \ref{pstC} again yields no PST between them. Despite this fact, we show that under certain assumptions, Laplacian PST can be induced in $X\vee Y$ by appropriate choice of $n$.

\begin{theorem}
\label{oddst}
Let $m\geq 2$ and $\phi(L(X),t)\in\mathbb{Z}[t]$. Suppose $u$ and $v$ are strongly cospectral vertices in $X$ with $\sigma_u(L)\subseteq\mathbb{Z}$ such that the $\nu_2(\lambda)$'s are equal for all $\lambda\in \sigma_u(L)$, say to $\alpha$. Write each $0\neq \lambda_r\in\sigma_{uv}^+(L(X))$ as $\lambda_r=2^{\alpha}(2p_r-1)$ and each $ \mu_s\in\sigma_{uv}^-(L(X))$ as $\mu_s=2^{\alpha}(2q_s-1)$ for some $p_r,q_s\in\mathbb{Z}$. Then Laplacian perfect state transfer occurs between $u$ and $v$ in $X\vee Y$ if and only if $X$ is connected, $m\notin \sigma_{uv}^-(L(X))$, $\nu_2(m)=\nu_2(n)=\alpha$ (so we may write $m=2^{\alpha}(2y-1)$ and $n=2^{\alpha}(2z+1)$ for some $y,z\in\mathbb{Z}$), and one of the following conditions below hold.
\begin{enumerate}
\item All $\nu_2(q_s)$'s are equal, $\nu_2(p_r)> \nu_2(q_s)$, and $\nu_2(y)\geq \nu_2(z)\geq \nu_2(q_s)+1$ for all $r,s$.
\item All $\nu_2(p_r)$'s are equal and $\nu_2(p_r)=\nu_2(y)=\nu_2(z)<\nu_2(q_s)$ for all $s$. 
\item For any $r,s,t$, we have $\nu_2(p_r)= \nu_2(q_s)=\nu_2(y)=\nu_2(z)$, $\nu_2\left(\frac{p_r+z}{2^{\nu_2(z)}}\right)>\nu_2\left(\frac{q_s+z}{2^{\nu_2(z)}}\right)=\nu_2\left(\frac{q_t+z}{2^{\nu_2(z)}}\right)$ and $\nu_2\left(\frac{y+z}{2^{\nu_2(z)}}\right)>\nu_2\left(\frac{q_s+z}{2^{\nu_2(z)}}\right)$.
\end{enumerate}
\end{theorem}

\begin{proof}
Corollary \ref{pstLc5} implies that PST occurs between $u$ and $v$ in $X\vee Y$ if and only if $\nu_2(m)=\nu_2(n)=\alpha$ and Theorem \ref{pstL}(3aii) holds. Let $m=2^{\alpha}(2y-1)$ and $n=2^{\alpha}(2z+1)$ for some integers $y,z$. Note that $m+n=2^{\alpha+1}(y+z)$, $\lambda_r+n=2^{\alpha+1}(p_r+z)$ and $\mu_s+n=2^{\alpha+1}(q_s+z)$. If $X$ is disconnected, then $n\in\sigma_{uv}^+(L(X))$. Now, Theorem \ref{pstL}(3aii) holds if and only if $X$ is connected and 
\begin{equation}
\label{contra}
\nu_2(p_r+z)>\nu_2(q_s+z)=\nu_2(q_t+z)\quad \text{and}\quad \nu_2(y+z)>\nu_2(q_s+z)
\end{equation}
for any $r,s,t$. One can then proceed by cases similar to the proof of Theorem \ref{pstL}(2a).
\end{proof}

\begin{remark}
\label{p3}
Let $X$ be a simple unweighted graph with strongly cospectral vertices $u$ and $v$ such that $\sigma_u(L)\subseteq\mathbb{Z}$. If $X$ has an odd number of vertices and an odd number of spanning trees, then the matrix-tree theorem implies that all nonzero eigenvalues in $\sigma_u(L)$ are odd. Thus, Theorem \ref{oddst} applies with $\alpha=0$.
\end{remark}

We illustrate Theorem \ref{oddst} and Remark \ref{p3} using a simple example.

\begin{example}
\label{p31}
Let $X=P_3$ with end vertices $u$ and $v$. Then $\sigma_{uv}^+(L(X))=\{0,3\}$, $\sigma_{uv}^-(L(X))=\{1\}$, and so by Theorem \ref{pstC}(2), $u$ and $v$ does not admit Laplacian PST in $X$. Applying Theorem \ref{oddst}(1), we get Laplacian PST between $u$ and $v$ in $X\vee Y$ if and only if $n\equiv 1$ (mod 4). In particular, taking $Y=O_n$ yields $X\vee Y=K_d\backslash e$ with $d\equiv 0$ (mod 4), a graph known to have PST between non-adjacent vertices.
\end{example}

We close this subsection with a result that determines all graphs with isolated vertices that exhibit Laplacian PST in the join. It is immediate from Theorems \ref{sc}(1a) and \ref{pstL}(3c).

\begin{theorem}
\label{genpstL}
If $u$ and $v$ are two isolated vertices in $X$, then Laplacian perfect state transfer occurs between $u$ and $v$ in $X\vee Y$ if and only if $X=O_2$ and $n\equiv 2$ (mod 4).
\end{theorem}

From Theorem \ref{genpstL}, $O_2$ is the only graph with isolated vertices that exhibits PST in the join. This coincides with a known characterization of Laplacian PST in double cones \cite[Corollary 4]{Alvir2016}. 

\subsection{Adjacency case}

For the adjacency case, we give one scenario where PST can be induced in $X\vee Y$. The following is immediate from Theorems \ref{pstC}(3) and \ref{pstA}(3).

\begin{corollary}
\label{pstAc5}
Let $m\geq 2$, $k,\ell\in\mathbb{Z}$ and $\phi(A(X),t)\in\mathbb{Z}[t]$.
Suppose $X$ is connected and vertices $u$ and $v$ are strongly cospectral in $X$ such that $\sigma_u(A(X))\subseteq\mathbb{Z}$ and
\begin{equation}
\label{adjnu21}
\nu_2\left(\lambda-\eta\right)>\nu_2\left(\lambda-\mu\right)=\nu_2\left(\lambda-\theta\right),
\end{equation}
for all $ \lambda,\eta\in \sigma_{uv}^+(A(X))\backslash\{k\}$ and $\mu,\theta\in\sigma_{uv}^-(A(X))$. If (\ref{adjnu21}) does not hold whenever $\lambda=k$, then
\begin{enumerate}
\item adjacency perfect state transfer does not occur between $u$ and $v$ in $X$; and
\item adjacency perfect state transfer occurs between $u$ and $v$ in $X\vee Y$ if and only if $\lambda^-\notin \sigma_{uv}^-(A(X))$, $D$ is a perfect square and (\ref{adjnu21}) also holds whenever $\lambda\in\{\lambda^{\pm}\}$.
\end{enumerate}
\end{corollary}

Using Corollary \ref{pstLc5}, we generate an infinite family of graphs with no PST but which admit PST in the join.

\begin{example}
\label{cpm1}

Let $m\geq 6$ and $X=CP(m)$, where $m\equiv 2$ (mod 4). From Example \ref{cpm}, $X$ is connected, $(m-2)$-regular and PST does not occur between $u$ and $v$ in $X$ by Corollary \ref{pstAc5}(1). Applying Corollary \ref{pstAc5}(2), we get adjacency PST between $u$ and $v$ in $X\vee Y$ if and only if $n=\frac{z(m-\ell-2+z)}{m}$ for some integer $z$ such that $m$ divides $z(m-\ell-2+z)$, $\nu_2(\ell)>\nu_2(z)=1$ and $\ell-z\neq 0$. The latter three conditions are equivalent to $D$ being a perfect square, $\nu_2(\lambda^+-\lambda^-)=\nu_2(\lambda^--0)=2$ and $\lambda^-\neq 0$, respectively. In particular, we may take $\ell=0$ and $z=2$ to get $n=2$. In this case, $X\vee Y=CP(m+2)$ which admits PST between any two non-adjacent vertices.
\end{example}

We finish this section with an analogue of Theorem \ref{genpstL}. It is immediate from Theorem \ref{pstA}.

\begin{theorem}
\label{genpstA}
If $u$ and $v$ are two isolated vertices in $X$, then adjacency perfect state transfer occurs between $u$ and $v$ in $X\vee Y$ if and only if $X=O_2(k)$, $D$ is a perfect square and $\nu_2(\lambda^+-k)=\nu_2(\lambda^--k)$.
\end{theorem}

\section{Families of joins}\label{secFam}

We now determine quantum state transfer in self-joins and iterated join graphs.

\subsection{Self-joins}

Denote the $r$-fold join of $X$ with itself by $X^r:=\bigvee_{j=1}^{r} X$. The following characterizes periodicity in $X^r$. It follows from Corollaries \ref{joinperpreserve} and \ref{perequiv}, and Theorem \ref{Xper}.

\begin{theorem}
Let $M\in\{A,L\}$ and $\phi(M(X),t)\in\mathbb{Z}[t]$. Vertex $u$ is periodic in $X$ if and only if it is periodic in $X^r$ for all $r$, if and only if $\sigma_u(M(X))\subseteq\mathbb{Z}$. Moreover, $X$ is periodic if and only if $X^r$ is periodic for all $r$, if and only if $X$ is (Laplacian) integral.
\end{theorem}

Next, we characterize strong cospectrality in $X^r$.

\begin{theorem}
\label{scr}
Suppose vertices $u$ and $v$ are strongly cospectral vertices in $X$ with respect to $M$.
\begin{enumerate}
\item If $m\notin\sigma_{uv}^-(L(X))$, then $u$ and $v$ are Laplacian strongly cospectral in $X^r$ for all $r\geq 1$ with
\begin{equation*}
\sigma_{uv}^+(L)=\{\lambda+(r-1)m: \lambda\in\sigma_{uv}^+(L(X))\backslash \{0\}\}\cup\mathcal{R}\quad \text{and}\quad \sigma_{uv}^-(L)=\{\mu+(r-1)m:\mu\in\sigma_{uv}^-(L(X))\}, 
\end{equation*}
where $\mathcal{R}=\{0,rm\}$ if $X$ is connected and $\mathcal{R}=\{0,rm,(r-1)m\}$ otherwise. If we add that $X\neq K_2$ is unweighted, then $u$ and $v$ are strongly cospectral in $X^r$ for all $r\geq 1$.
\item If $k-m\notin\sigma_{uv}^-(A(X))$, then $u$ and $v$ are adjacency strongly cospectral in $X^r$ for all $r\geq 1$ with
\begin{equation*}
\sigma_{uv}^+(A)=\sigma_{uv}^+(A(X))\backslash\{k\}\cup\mathcal{R} \quad \text{and}\quad \sigma_{uv}^-(A)=\sigma_{uv}^-(A(X)), 
\end{equation*}
where $\mathcal{R}= \{k-m,k+(r-1)m\}$ if $X$ is connected and $\mathcal{R}= \{k-m,k+(r-1)m,k\}$ otherwise.
\end{enumerate}
\end{theorem}

\begin{proof}
Note that (1) follows from Theorem \ref{sc}(1) and Remark \ref{scj}. To prove (2), note that $X^r$ is $(k+(r-1)m)$-regular with $rm$ vertices. Taking $Y=X^{r-1}$, we get $X^{r}=X\vee Y$ and $D=((k+(r-2)m)-k)^2+4(r-1)m^2=r^2m^2$, and so $\lambda^-=k-m$ and $\lambda^+=k+(r-1)m$. Invoking Theorem \ref{sc}(2) yields (2).
\end{proof}

Taking $Y=X^{r-1}$ in Theorem \ref{pstL} yields a characterization of Laplacian PST in $r$-fold joins.

\begin{theorem}
\label{pstAxr}
Let $r,m\geq 2$ and $\phi(L(X),t)\in\mathbb{Z}[t]$. Vertices $u$ and $v$ in $X$ admit Laplacian perfect state transfer in $X^r$ if and only if either:
\begin{enumerate}
\item $X=O_2$, in which case $X^r=CP(2r)$, where $r$ is even.
\item Vertices $u$ and $v$ belong to the same connected component of $X$ and all of the conditions below hold.
\begin{enumerate}
\item Vertices $u$ and $v$ are Laplacian strongly cospectral in $X$ and $m\notin\sigma_{uv}^-(L(X))$.
\item The eigenvalues in $\sigma_u(L(X))$ are all integers.
\item One of the following conditions hold.
\begin{enumerate}
\item $X$ is connected and one of the following conditions hold for all $0\neq \lambda\in\sigma_{uv}^+(L(X))\cup\{m\}$ and $\mu,\theta\in\sigma_{uv}^-(L(X))$.
\begin{enumerate}
\item $\nu_2(\lambda)>\nu_2(\mu)=\nu_2(\theta)$ and $r$ is any integer.
\item $\nu_2(\mu)>\nu_2(\lambda)$ and $r$ is even.
\item $\nu_2(\lambda)=\nu_2(\mu)$, $\nu_2\left(\frac{\lambda+(r-1)m}{2^{\nu_2(m)}}\right)>\nu_2\left(\frac{\mu+(r-1)m}{2^{\nu_2(m)}}\right)=\nu_2\left(\frac{\theta+(r-1)m}{2^{\nu_2(m)}}\right)$ and $r$ is even.
\end{enumerate}
\item $X\neq O_2$ is disconnected and (iA) above holds.
\end{enumerate}
\end{enumerate}
\end{enumerate}
If these conditions hold, then $\tau_{X^r}=\frac{\pi}{g}$, where $g=\operatorname{gcd}\left(\mathcal{T}\right)$, $\mathcal{T}=\{rm\}\cup \{\lambda-m:\lambda\in \sigma_u(L(X))\backslash\{0\}\}$ if $X$ is connected, and $\mathcal{T}=\{m\}\cup\sigma_u(L(X))\backslash\{0\}$ otherwise.
\end{theorem}

\begin{example}
Let $X=P_3$ with end vertices $u$ and $v$, and consider $X^r$ for any integer $r$. From Example \ref{p31}, $\sigma_{uv}^+(L(X))=\{0,3\}$, $\sigma_{uv}^-(L(X))=\{1\}$ and Laplacian PST does not occur in $X$. Since $\lambda=m=3$ and $\mu=1$ are odd, we have $\nu_2(\lambda+(r-1)m)=\nu_2(r)$ and $\nu_2(\mu+(r-1)m)=1$. By Theorem \ref{pstAxr}(2ci), PST occurs between $u$ and $v$ in $X^r$ if and only if $r\equiv 0$ (mod 4). In this case, $P_3^r=(O_2\vee O_1)^r=CP(2r)\vee K_r$. Indeed, Theorem \ref{pstAxr}(1) yields PST in $CP(2r)$ between non-adjacent vertices with minimum time $\frac{\pi}{2}$, and as $r\equiv 0$ (mod 4), Corollary \ref{pstLc2} implies that this PST is preserved in $CP(2r)\vee K_r$ at the same time.
\end{example}

Using Corollary \ref{pstLc2}, we characterize when Laplacian PST in preserved in self-joins.

\begin{corollary}
\label{pstAxrcor}
Let $m\geq 2$ and $\phi(L(X),t)\in\mathbb{Z}[t]$. Suppose Laplacian perfect state transfer occurs between vertices $u$ and $v$ in $X$ and $\tau_X=\frac{\pi}{h}$. Then it occurs between them in $X^r$ if and only if $m\notin\sigma_{uv}^-(L(X))$ and $m\equiv 0$ (mod $2^{\alpha}$), where $\alpha>\nu_2(h)$. In this case, $\tau_{X^r}=\frac{\pi}{g}$, where $g=\operatorname{gcd}(\mathcal{T})$, $\mathcal{T}$ is given in Theorem \ref{pstAxr} and $\nu_2(g)=\nu_2(h)$. Moreover, if $X$ is connected, then $\frac{h}{g}$ is rational. Otherwise, $\frac{h}{g}=\frac{h}{\operatorname{gcd}(h,m)}$ is odd.
\end{corollary}

In Theorem \ref{pstAxr}(2a) and Corollary \ref{pstAxrcor}, we may drop the condition $m\notin\sigma_{uv}^-(L(X))$ if $X\neq K_2$ is unweighted. Taking $Y=X^{r-1}$ in Theorem \ref{pstA} yields a characterization of adjacency PST in $r$-fold joins. 

\begin{theorem}
\label{pstLxr}
Let $r,m\geq 2$, $k\in\mathbb{Z}$ and $\phi(A(X),t)\in\mathbb{Z}[t]$. If $X$ is simple, then vertices $u$ and $v$ in $X$ admit adjacency perfect state transfer in $X^r$ if and only if either:
\begin{enumerate}
\item $X=O_2$, in which case $X^r=CP(2r)$, where $r$ is even.
\item Vertices $u$ and $v$ belong to the same connected component of $X$ and all of the conditions below hold.
\begin{enumerate}
\item Vertices $u$ and $v$ are adjacency strongly cospectral in $X$ and $k-m\notin\sigma_{uv}^-(L(X))$.
\item The eigenvalues in $\sigma_u(A(X))$ are all integers.
\item One of the following conditions hold.
\begin{enumerate}
\item $X$ is connected and one of the following conditions hold for all $0\neq \lambda\in\sigma_{uv}^+(A(X))\backslash\{k\}$ and $\mu,\theta\in\sigma_{uv}^-(A(X))$.
\begin{enumerate}
\item $\nu_2(m)>\nu_2(k-\mu)$, $\nu_2(k-\lambda)>\nu_2(k-\mu)$ and $r$ is any integer.
\item $\nu_2(k-\mu)>\nu_2(m)=\nu_2(k-\lambda)$ and $r$ is even
\item $\nu_2(m)=\nu_2(k-\eta)=\nu_2(k-\mu)$, $\nu_2\left(\frac{k-m-\lambda}{2^{\nu_2(m)}}\right)>\nu_2\left(\frac{k-m-\mu}{2^{\nu_2(m)}}\right) =\nu_2\left(\frac{k-m-\theta}{2^{\nu_2(m)}}\right)$ and $r$ is an integer such that $\nu_2(r)>\nu_2\left(\frac{k-m-\mu}{2^{\nu_2(m)}}\right)$.
\end{enumerate}
\item $X\neq O_2$ is disconnected and (iA) above holds.
\end{enumerate}

\end{enumerate}
\end{enumerate}
If these conditions hold, then  $\tau_{X^r}=\frac{\pi}{g}$, where $g=\operatorname{gcd}\left(\mathcal{T}\right)$, $\mathcal{T}=\{rm\}\cup \{k-m-\lambda:\lambda\in\sigma_u(A(X))\backslash\{k\}\}$ if $X$ is connected and $\mathcal{T}=\{m\}\cup \{k-\lambda:\lambda\in\sigma_u(A(X))\backslash\{k\}\}$ otherwise.
\end{theorem}

\begin{example}
Let $X=CP(m)$, where $m\equiv 2$ (mod 4). From Example \ref{cpm}, $\sigma_{uv}^{+}(A(X))=\{m-2,-2\}$, $\sigma_{uv}^{-}(A(X))=\{0\}$ and adjacency PST does not occur between non-adjacent vertices $u$ and $v$. If we let $k=m-2$, $\lambda=-2$ and $\mu=0$, then $\nu_2(k-\mu)>\nu_2(m)=\nu_2(k-\lambda)$. By Theorem \ref{pstLxr}(2ciB), adjacency PST occurs between $u$ and $v$ in $X^r$ if and only if $r$ is even. In this case, $X^r=CP(m)^r=CP(rm)$. Since $rm=2r(m/2)$, we may also use Theorem \ref{pstLxr}(1) to conclude that $X^r$ has PST for all even $r$.
\end{example}

We also characterize when adjacency PST in preserved in self-joins using Corollary \ref{adjpstinX}.

\begin{corollary}
Let $r,m\geq 2$, $k\in\mathbb{Z}$ and $\phi(A(X),t)\in\mathbb{Z}[t]$. Suppose $X$ is simple and adjacency perfect state transfer occurs between vertices $u$ and $v$ in $X$ and $\tau_X=\frac{\pi}{h}$. Then it occurs between them in $X^r$ if and only if  $k-m\notin\sigma_{uv}^{-}(A(X))$ and $\nu_2(m)>\nu_2(k-\mu)$ for all $\mu\in \sigma_{uv}^{-}(A(X))$. Here, $\tau_{X^r}=\frac{\pi}{g}$, where $g=\operatorname{gcd}(\mathcal{T})$, $\mathcal{T}$ is given in Theorem \ref{pstLxr} and $\nu_2(g)=\nu_2(h)$. Further, if $X$ is connected, $\frac{h}{g}$ is rational. Otherwise, $\frac{h}{g}=\frac{h}{\operatorname{gcd}(h,m)}$ is odd.
\end{corollary}

\subsection{Iterated join graphs}

In this subsection, we characterize Laplacian strong cospectrality and PST in iterated join graphs.

An \textit{iterated join graph} is a weighted graph that is either of the form:
\begin{itemize}
\item[(i)] $((((X_{m_1}\vee X_{m_2})\cup X_{m_3})\vee X_{m_4})\ldots)\vee X_{m_{2k}}:=\Gamma(X_{m_1},\ldots,X_{m_{2k}})$
\item[(ii)] $((((X_{m_1}\cup X_{m_2})\vee X_{m_3})\cup X_{m_4})\ldots)\vee X_{m_{2k+1}}:=\Gamma(X_{m_1},\ldots,X_{m_{2k+1}})$,
\end{itemize}
where $X_{m_1},\ldots,X_{m_{2r}}$ are weighted graphs on $m_1,\ldots,m_{2r}\geq 1$ vertices, respectively.

If we take $X_j=O_{m_j}$ for odd $j$ and $X_j=K_{m_j}$ otherwise, then $\Gamma(O_{m_1},K_{m_2}\ldots,K_{m_{2k}})$ is a connected threshold graph. $\Gamma(K_{m_1},O_{m_2},\ldots,K_{m_{2k+1}})$ is also a threshold graph,  where $X_j=O_{m_j}$ for even $j$ and $X_j=K_{m_j}$ otherwise. In fact, graphs of either form completely describe the class of connected threshold graphs \cite{Severini}. Thus, iterated join graphs are generalizations of connected theshold graphs.

Lemma \ref{supp} and the fact that the union preserves eigenvalue supports yield the following result which determines the eigenvalue support of each vertex in an iterated join graph.

\begin{proposition}
\label{ijgsupp}
The following hold.
\begin{enumerate}
\item Let $X=\Gamma(X_{m_1},\ldots,X_{m_{2k}})$ and suppose $u\in V(X_{m_j})$ for some $j\in\{1,\ldots,2k\}$. Define $\alpha_h=\sum_{\ell=1}^hm_{\ell}$, $\beta_h=\sum_{\ell=1}^{(2k-h)/2}m_{h+2\ell}$ when $h\neq 2k$ and $\beta_{2k}=0$. The following hold.
\begin{enumerate}
\item If $j$ is odd, then $\sigma_{u}(L(X))=\{\lambda+\beta_{j-1}:\lambda\in\sigma_u(L(X_{m_j}))\backslash \{0\}\}\cup \{\alpha_h+\beta_h:j\leq h\leq 2k\ \text{is even}\}\cup \mathcal{R}$, where $\mathcal{R}=\{0\}$ whenever $j=1$ and $X_{m_1}$ is connected, and $\mathcal{R}=\{0,\beta_{j-1}\}$ otherwise.
\item If $j$ is even, then $\sigma_{u}(L(X))=\{\lambda+\alpha_{j-1}+\beta_{j}:\lambda\in\sigma_u(L(X_{m_j}))\backslash \{0\}\}\cup \{\alpha_h+\beta_h:j\leq h\leq 2k\ \text{is even}\}\cup \mathcal{R}$, where $\mathcal{R}=\{0\}$ whenever $X_j$ is connected and $\mathcal{R}=\{0,\alpha_{j-1}+\beta_{j}\}$ otherwise.
\end{enumerate}
\item Let $X=\Gamma(X_{m_1},\ldots,X_{m_{2k+1}})$ and suppose $u\in V(X_{m_j})$ for some $j\in\{1,\ldots,2k+1\}$. Define $\alpha_h=\sum_{\ell=1}^hm_{\ell}$, $\beta_h=\sum_{\ell=1}^{(2k+1-h)/2}m_{h+2\ell}$ when $h\neq 2k+1$ and $\beta_{2k+1}=0$. The following hold.
\begin{enumerate}
\item If $j=1$, then $\sigma_{u}(L(X))=\{\lambda+\beta_1:\lambda\in\sigma_u(L(X_{m_j}))\backslash \{0\}\}\cup \{\alpha_{h+1}+\beta_h:1\leq h\leq 2k+1\ \text{is odd}\}\cup \mathcal{R}$, where $\mathcal{R}=\{0,\beta_1\}$.
\item If $j\geq 3$ is odd, then $\sigma_{u}(L(X))=\{\lambda+\alpha_{j-1}+\beta_j:\lambda\in\sigma_u(L(X_{m_j}))\backslash \{0\}\}\cup \{\alpha_{h}+\beta_h:j\leq h\leq 2k+1\ \text{is odd}\}\cup \mathcal{R}$, where $\mathcal{R}=\{0\}$ if $X$ is connected and $\mathcal{R}=\{0,\alpha_{j-1}+\beta_j\}$ otherwise.
\item If $j$ is even, then $\sigma_{u}(L(X))=\{\lambda+\beta_{j-1}:\lambda\in\sigma_u(L(X_{m_j}))\backslash \{0\}\}\cup \{\alpha_{h}+\beta_h:j+1\leq h\leq 2k+1\ \text{is odd}\}\cup \mathcal{R}$, where $\mathcal{R}=\{0,\beta_{j-1}\}$.
\end{enumerate}
\end{enumerate}
\end{proposition}

Using Proposition \ref{ijgsupp} and the same argument in the proof of Theorem \ref{sc} yield the following characterization of Laplacing strong cospectrality in iterated join graphs.

\begin{theorem}
\label{ijgsc}
Let $X=\Gamma(X_{m_1},\ldots,X_{m_{2k}})$ and consider $\alpha_h$ and $\beta_h$ in Proposition \ref{ijgsupp}(1). Two vertices $u$ and $v$ of $X_{m_j}$ are Laplacian strongly cospectral in $X$ if and only if one of the following conditions hold.
\begin{enumerate}
\item $j$ is odd and either:
\begin{enumerate}
\item Vertices $u$ and $v$ are Laplacian strongly cospectral in $X_{m_j}$ and $\alpha_j\notin\sigma_{uv}^-(L(X_{m_j}))$. In this case, $\sigma_{uv}^+(L)=\{\lambda+\beta_{j-1}:\lambda\in\sigma_{uv}^+(L(X_{m_j}))\backslash\{0\}\}\cup \{\alpha_h+\beta_h:j\leq h\leq 2k\ \text{is even}\}\cup \mathcal{R}$ and $\sigma_{uv}^-(L)=\{\mu+\beta_{j-1}:\mu\in\sigma_{uv}^-(L(X_{m_j}))\}$, where $\mathcal{R}$ is given in Proposition \ref{ijgsupp}(1a).
\item $j=1$ and $X_{m_1}=O_2$ with $\sigma_{uv}^+(L)=\{\alpha_h+\beta_h:1\leq h\leq 2k\ \text{is even}\}\cup \{0\}$ and $\sigma_{uv}^-(L)=\{\beta_0\}$.
\end{enumerate}
\item $j$ is even and either:
\begin{enumerate}
\item Vertices $u$ and $v$ are Laplacian strongly cospectral in $X_{m_j}$ and $m_j\notin\sigma_{uv}^-(L(X_{m_j}))$. In this case, $\sigma_{uv}^+(L)=\{\lambda+\alpha_{j-1}+\beta_j:\lambda\in\sigma_{uv}^+(L(X_{m_j}))\backslash\{0\}\}\cup \{\alpha_h+\beta_h:j\leq h\leq 2k\ \text{is even}\}\cup \mathcal{R}$ and $\sigma_{uv}^-(L)=\{\mu+\alpha_{j-1}+\beta_j:\mu\in\sigma_{uv}^-(L(X_{m_j}))\}$, where $\mathcal{R}$ is given in Proposition \ref{ijgsupp}(1b).
\item $X_{m_j}=O_2$ with $\sigma_{uv}^+(L)=\{\alpha_h+\beta_h:j\leq h\leq 2k\ \text{is even}\}\cup  \{0\}$ and $\sigma_{uv}^-(L)=\{\alpha_{j-1}+\beta_j\}$.
\end{enumerate}
\end{enumerate}
\end{theorem}

\begin{theorem}
\label{ijgsc1}
Let $X=\Gamma(X_{m_1},\ldots,X_{m_{2k+1}})$ and consider $\alpha_h$, $\beta_h$ and $\mathcal{R}$ in Proposition \ref{ijgsupp}(2). Two vertices $u$ and $v$ of $X_{m_j}$ are Laplacian strongly cospectral in $X$ if and only if one of the following conditions hold.
\begin{enumerate}
\item $j=1$, vertices $u$ and $v$ are Laplacian strongly cospectral in $X_{m_1}$ and $m_1+m_2\notin\sigma_{uv}^-(L(X_{m_1}))$. In this case, $\sigma_{uv}^+(L)=\{ \lambda+\beta_1:\lambda\in\sigma_{uv}^+(L(X_{m_1}))\backslash\{0\}\}\cup \{\alpha_{h+1}+\beta_h:1\leq h\leq 2k+1\ \text{is odd}\}\cup \mathcal{R}$
and $\sigma_{uv}^-(L)=\{\mu+\beta_1:\mu\in\sigma_{uv}^-(L(X_{m_1}))\}$, where $\mathcal{R}$ is given in Proposition \ref{ijgsupp}(2a).
\item $j\geq 3$ is odd and either:
\begin{enumerate}
\item Vertices $u$ and $v$ are Laplacian strongly cospectral in $X_{m_j}$ and $m_j\notin\sigma_{uv}^-(L(X_{m_j}))$. In this case, $\sigma_{uv}^+(L)=\{\lambda+\alpha_{j-1}+\beta_j:\lambda\in\sigma_{uv}^+(L(X_{m_j}))\backslash\{0\}\}\cup \{\alpha_h+\beta_h:j\leq h\leq 2k+1\ \text{is odd}\}\cup \mathcal{R}$ and $\sigma_{uv}^-(L)=\{\mu+\alpha_{j-1}+\beta_j:\mu\in\sigma_{uv}^-(L(X_{m_j}))\}$, where $\mathcal{R}$ is given in Proposition \ref{ijgsupp}(2b).
\item $X_{m_j}=O_2$ with $\sigma_{uv}^+(L)=\{\alpha_h+\beta_h:j\leq h\leq 2k\ \text{is odd}\}\cup  \{0\}$ and $\sigma_{uv}^-(L)=\{\alpha_{j-1}+\beta_j\}$.
\end{enumerate}
\item $j$ is even, vertices $u$ and $v$ are Laplacian strongly cospectral in $X_{m_j}$, and $\alpha_j\notin \sigma_{uv}^-(L(X_{m_j}))$. In this case,
$\sigma_{uv}^+(L)=\{\lambda+\beta_{j-1}:\lambda\in\sigma_{uv}^+(L(X_{m_j}))\backslash\{0\}\}\cup \{\alpha_h+\beta_h:j\leq h\leq 2k+1\ \text{is odd}\}\cup \mathcal{R}$ and $\sigma_{uv}^-(L)=\{\mu+\beta_{j-1}:\mu\in\sigma_{uv}^-(L(X_{m_j}))\}$, where $\mathcal{R}$ is given in Proposition \ref{ijgsupp}(2c).
\end{enumerate}
\end{theorem}

\begin{remark}
From Theorems \ref{ijgsc} and \ref{ijgsc1}, a pair of strongly cospectral vertices in an iterated join graph belong to the same $X_{m_j}$. Consequently, if $u$ and $v$ belong to $X_{m_j}$ and $X_{m_{\ell}}$ respectively, then PST cannot occur between them.
\end{remark}

Our next endeavour is to characterize PST in iterated join graphs, starting with those of the form $\Gamma(X_{m_1},\ldots,X_{m_{2k}})$. We separate the cases when $j$ is odd and $j$ is even. The following result follows from combining Corollary \ref{joinperpreserve}, Theorem \ref{ijgsc}, the same argument in the proof of Theorem \ref{pstL}, and the fact that
\begin{equation}
\label{gam}
\alpha_h+\beta_h=\alpha_j+\gamma_h+\beta_{j-1}
\end{equation}
for a fixed odd $j$ and for all even $j\leq h\leq 2k$ in Theorem \ref{ijgsc}(1) with $\gamma_h=0$ whenever $h=j+1$ and $\gamma_h=m_{j+2}+m_{j+4}++m_{j+6}\ldots+m_{h-1}$ otherwise, and
\begin{equation}
\label{del}
\alpha_h+\beta_h=\alpha_{j-1}+\delta_h+\beta_{j}
\end{equation}
for a fixed even $j$ and for all even $j\leq h\leq 2k$ in Theorem \ref{ijgsc}(2) with $\delta_h=m_{j}+m_{j+1}+m_{j+3}+\ldots+m_{h-1}$.

\begin{theorem}
\label{pstLijg}
Let $X=\Gamma(X_{m_1},\ldots,X_{m_{2k}})$, $\phi(L(X),t)\in\mathbb{Z}[t]$ and $m_j\geq 2$ for some odd $j$. Consider $\alpha_h$ and $\beta_h$ in Proposition \ref{ijgsupp}(1) and $\gamma_h$ in (\ref{gam}). Two vertices $u$ and $v$ in $X_{m_j}$ admit Laplacian perfect state transfer in $X$ if and only if all of the following conditions hold.
\begin{enumerate}
\item Either (i) $u$ and $v$ are Laplacian strongly cospectral in $X_{m_j}$ and $\alpha_j\notin\sigma_{uv}^-(L(X_{m_j}))$ or (ii) $j=1$ and $X_{m_1}=O_2$.
\item The eigenvalues in $\sigma_u(L(X_{m_j}))$ are all integers.
\item One of the following conditions hold.
\begin{enumerate}
\item $X_{m_j}$ is connected and one of the following conditions hold for all $\lambda\in\sigma_{uv}^+(L(X_{m_j}))\backslash\{0\}$, $\mu,\theta\in\sigma_{uv}^-(L(X_{m_j}))$ and even $j\leq h\leq 2k$.
\begin{enumerate}
\item $\nu_2(\beta_{j-1})>\nu_2(\mu)=\nu_2(\theta)$, $\nu_2(\lambda)>\nu_2(\mu)$ and $ \nu_2(\alpha_j+\gamma_h)>\nu_2(\mu)$.
\item $\nu_2(\mu)>\nu_2(\beta_{j-1})=\nu_2(\lambda)= \nu_2(\alpha_j+\gamma_h)$.
\item $\nu_2(\beta_{j-1})=\nu_2(\mu)=\nu_2(\lambda)= \nu_2(\alpha_j+\gamma_h)$, $\nu_2\left(\frac{\lambda+\beta_{j-1}}{2^{\nu_2(\beta_{j-1})}}\right)>\nu_2\left(\frac{\mu+\beta_{j-1}}{2^{\nu_2(\beta_{j-1})}}\right)=\nu_2\left(\frac{\theta+\beta_{j-1}}{2^{\nu_2(\beta_{j-1})}}\right)$ and $\nu_2\left(\frac{\alpha_h+\beta_h}{2^{\nu_2(\beta_{j-1})}}\right)>\nu_2\left(\frac{\mu+\beta_{j-1}}{2^{\nu_2(\beta_{j-1})}}\right).$
\end{enumerate}
\item $X_{m_j}$ is disconnected and (ai) above holds.
\item $j=1$, $X_{m_1}=O_2$ and $\nu_2(2+\gamma_h+\beta_0)>\nu_2(\beta_0)$ for each even $1\leq h\leq 2k$.
\end{enumerate}
\end{enumerate}
Further, $\tau_{X\vee Y}=\frac{\pi}{g}$, where $g=\operatorname{gcd}\left(\mathcal{T}\right)$, where $\mathcal{T}=\sigma_u(L(X))$ is given in Proposition \ref{ijgsupp}(1).
\end{theorem}

\begin{theorem}
\label{pstLijg1}
Let $X=\Gamma(X_{m_1},\ldots,X_{m_{2k}})$, $\phi(L(X),t)\in\mathbb{Z}[t]$ and $m_j\geq 2$ for some even $j$. Consider $\alpha_h$ and $\beta_h$ in Proposition \ref{ijgsupp}(1) and $\delta_h$ in (\ref{del}) Two vertices $u$ and $v$ in $X_{m_j}$ admit Laplacian perfect state transfer in $X$ if and only if all of the following conditions hold.
\begin{enumerate}
\item Either (i) $u$ and $v$ are Laplacian strongly cospectral in $X_{m_j}$ and $m_j\notin\sigma_{uv}^-(L(X_{m_j}))$ or (ii) $X_{m_j}=O_2$.
\item The eigenvalues in $\sigma_u(L(X_{m_j}))$ are all integers.
\item One of the following conditions hold.
\begin{enumerate}
\item $X_{m_j}$ is connected and one of the following conditions hold for all $\lambda\in\sigma_{uv}^+(L(X_{m_j}))\backslash\{0\}$, $\mu,\theta\in\sigma_{uv}^-(L(X_{m_j}))$ and even $j\leq h\leq 2k$.
\begin{enumerate}
\item $\nu_2(\alpha_{j-1}+\beta_j)>\nu_2(\mu)=\nu_2(\theta)$, $\nu_2(\lambda)>\nu_2(\mu)$ and $ \nu_2(\delta_h)>\nu_2(\mu)$.
\item $\nu_2(\mu)>\nu_2(\alpha_{j-1}+\beta_j)=\nu_2(\lambda)= \nu_2(\delta_h)$.
\item $\nu_2(\alpha_{j-1}+\beta_j)=\nu_2(\mu)=\nu_2(\lambda)=\nu_2(\delta_h)$, $\nu_2\left(\frac{\lambda+\alpha_{j-1}+\beta_j}{2^{\nu_2(\alpha_{j-1}+\beta_j)}}\right)>\nu_2\left(\frac{\mu+\alpha_{j-1}+\beta_j}{2^{\nu_2(\alpha_{j-1}+\beta_j)}}\right)=\nu_2\left(\frac{\theta+\alpha_{j-1}+\beta_j}{2^{\nu_2(\alpha_{j-1}+\beta_j)}}\right)$ and $\nu_2\left(\frac{\alpha_h+\beta_h}{2^{\nu_2(\alpha_{j-1}+\beta_j)}}\right)>\nu_2\left(\frac{\mu+\alpha_{j-1}+\beta_j}{2^{\nu_2(\alpha_{j-1}+\beta_j)}}\right).$
\end{enumerate}
\item $X_{m_j}$ is disconnected and (ai) above holds.
\item $X_{m_j}=O_2$ and $\nu_2(\alpha_{j-1}+\delta_h+\beta_j)>\nu_2(\alpha_{j-1}+\beta_j)$ for each even $j\leq h\leq 2k$.
\end{enumerate}
\end{enumerate}
Further, $\tau_{X\vee Y}=\frac{\pi}{g}$, where $g=\operatorname{gcd}\left(\mathcal{T}\right)$, where $\mathcal{T}=\sigma_u(L(X))$ is given in Proposition \ref{ijgsupp}(1).
\end{theorem}

\begin{remark}
\label{one}
We have the following.
\begin{enumerate}
\item We obtain a characterization for PST between vertices in $X_{m_j}$ in the graph $X=\Gamma(X_{m_1},\ldots,X_{m_{2k+1}})$ whenever $m_j\geq 2$ for some even $j\geq 2$ by replacing `even $j\leq h\leq 2k$' with `odd $j\leq h\leq 2k+1$' in Theorem \ref{pstLijg}.
\item We obtain a characterization for PST  between vertices in $X_{m_j}$ in the graph $X=\Gamma(X_{m_1},\ldots,X_{m_{2k+1}})$ whenever $m_j\geq 2$ for some odd $j\geq 3$ by replacing `even $j\leq h\leq 2k$' with `odd $j\leq h\leq 2k+1$' in Theorem \ref{pstLijg1}.
\end{enumerate}
\end{remark} 

In order to complete our characterization of PST in iterated join graphs, it suffices to determine PST between vertices in $X_{m_1}$ in the graph $X=\Gamma(X_{m_1},\ldots,X_{m_{2k+1}})$ by virtue of Remark \ref{one}. To do this, we combine Corollary \ref{joinperpreserve}, Theorem \ref{ijgsc}, the same argument in the proof of Theorem \ref{pstL}, and the fact that
\begin{equation}
\label{phi}
\alpha_{h+1}+\beta_h=\phi_h+\beta_1
\end{equation}
for all odd $1\leq h\leq 2k+1$ in Theorem \ref{ijgsc1}(1)
for a fixed odd $j$ with $\phi_h=m_1+m_2+m_4+m_6+\ldots+m_{h-1}$.

\begin{theorem}
\label{pstLijg2}
Let $X=\Gamma(X_{m_1},\ldots,X_{m_{2k+1}})$, $m_1\geq 2$ and $\phi(L(X),t)\in\mathbb{Z}[t]$. Consider $\alpha_h$ and $\beta_h$ in Proposition \ref{ijgsupp}(2) and $\phi_h$ in (\ref{phi}). Two vertices $u$ and $v$ in $X_{m_1}$ admit Laplacian perfect state transfer in $X$ if and only if all of the following conditions hold.
\begin{enumerate}
\item Vertices $u$ and $v$ are Laplacian strongly cospectral in $X_{m_1}$ and $m_1+m_2\notin\sigma_{uv}^-(L(X_{m_1}))$.
\item The eigenvalues in $\sigma_u(L(X_{m_1}))$ are all integers.
\item One of the following conditions hold.
\begin{enumerate}
\item $X_{m_1}$ is connected and one of the following conditions hold for all $\lambda\in\sigma_{uv}^+(L(X_{m_1}))\backslash\{0\}$, $\mu,\theta\in\sigma_{uv}^-(L(X_{m_1}))$ and odd $1\leq h\leq 2k+1$.
\begin{enumerate}
\item $\nu_2(\beta_1)>\nu_2(\mu)=\nu_2(\theta)$, $\nu_2(\lambda)>\nu_2(\mu)$ and $\nu_2(\phi_h)>\nu_2(\mu)$.
\item $\nu_2(\mu)>\nu_2(\beta_1)=\nu_2(\lambda)= \nu_2(\phi_h)$.
\item $\nu_2(\beta_1)=\nu_2(\mu)=\nu_2(\lambda)= \nu_2(\phi_h)$, $\nu_2\left(\frac{\lambda+\beta_1}{2^{\nu_2(\beta_1)}}\right)>\nu_2\left(\frac{\mu+\beta_1}{2^{\nu_2(\beta_1)}}\right)=\nu_2\left(\frac{\theta+\beta_1}{2^{\nu_2(\beta_1)}}\right)$ and $\nu_2\left(\frac{\alpha_h+\beta_h}{2^{\nu_2(\beta_1)}}\right)>\nu_2\left(\frac{\mu+\beta_1}{2^{\nu_2(\beta_1)}}\right).$
\end{enumerate}
\item $X_{m_1}$ is disconnected and (ai) above holds.
\end{enumerate}
In both cases, $u$ and $v$ belong to the same component in $X_{m_1}$.
\end{enumerate}
Further, $\tau_{X\vee Y}=\frac{\pi}{g}$, where $g=\operatorname{gcd}\left(\mathcal{T}\right)$, where $\mathcal{T}=\sigma_u(L(X))$ is given in Proposition \ref{ijgsupp}(2).
\end{theorem}

\begin{remark}
If $X_{m_j}\neq K_2$ is unweighted, then we may drop the condition $m_1+m_2\in\sigma_{uv}^-(L(X_{m_1}))$ in Theorem \ref{ijgsc1}(1), the condition $\alpha_j\notin\sigma_{uv}^-(L(X_{m_j}))$ in Theorems \ref{ijgsc}(1a) and \ref{ijgsc1}(3), and the condition $m_j\notin\sigma_{uv}^-(L(X_{m_j}))$ in Theorems \ref{ijgsc}(2a) and \ref{ijgsc1}(2a). The same applies to Theorems \ref{pstLijg}, \ref{pstLijg1} and \ref{pstLijg2}.
\end{remark}

\begin{remark}
\label{lastrem}
The conditions in Theorem \ref{pstLijg}(3c) are equivalent to $m_1=2$ and
\begin{equation*}
\nu_2(2+\beta_0)>\nu_2(\beta_0),\quad \nu_2(2+m_3+\beta_0)>\nu_2(\beta_0),\quad\ldots \quad \nu_2(2+m_3+m_5+\ldots+m_{2k-1}+\beta_0)>\nu_2(\beta_0).
\end{equation*}
The first inequality above is equivalent to $\beta_0\equiv 2$ (mod 4), and the latter ones yield $m_j\equiv 0$ (mod 4) for all odd $j\geq 3$. Since $X=Z_1\vee X_{m_{2k}}$ for some disconnected graph $Z_1$, applying Theorem \ref{pstL}(3ab) yields $\nu_2(m_{2k})>\nu_2(\beta_0)$ and
$\nu_2(m_1+m_2+m_4+\ldots+m_{2k-2})>\nu_2(\beta_0)$. Now, Corollary \ref{pstLc4} implies that PST also occurs between $u$ and $v$ in $Z_2$, where $Z_2$ is the graph such that $Z_1=Z_2\cup X_{m_{2k-1}}$. The same argument applied $k-2$ times yields $\nu_2(m_j)>\nu_2(\beta_0)$ for all even $j\geq 3$ and PST between $u$ and $v$ in $Z_{m_{k-2}}=O_{2}\vee X_{m_2}$. As $\beta_0\equiv 2$ (mod 4), the former condition is equivalent to $m_j\equiv 0$ (mod 4) for all even $j\geq 3$, while the latter is equivalent to $m_2\equiv 2$ (mod 4) by virtue of Theorem \ref{pstL}(3c). These considerations imply that Theorem \ref{pstLijg}(3c) holds if and only if $m_1=2$, $m_2\equiv 2$ (mod 4) and $m_j\equiv 0$ (mod 4) for all $j\geq 3$. One can also check that Theorem \ref{pstLijg2} holds for $X_{m_1}=K_2$ if and only if $m_1=2$, $m_2\equiv 2$ (mod 4) and $m_j\equiv 0$ (mod 4) for all $j\geq 3$. In both cases, the minimum PST time is $\frac{\pi}{2}$.
\end{remark}

\begin{corollary}
Let $X$ be a connected threshold graph. Then Laplacian perfect state transfer occurs in $X$ if and only if $m_1=2$, $m_2\equiv 2$ (mod 4) and $m_j\equiv 2$ (mod 4) for all $j\geq 3$. In this case, Laplacian perfect state transfer occurs between vertices $u$ and $v$ in $X_{m_1}\in\{O_{2},K_{2}\}$ with minimum PST time is $\frac{\pi}{2}$.
\end{corollary}

\begin{proof}
By assumption, $X\in\{\Gamma(O_{m_1},K_{m_1},\ldots,K_{m_{2k}}),\Gamma(K_{m_1},O_{m_1},\ldots,K_{m_{2k+1}})\}$. First, suppose $X=\Gamma(O_{m_1},K_{m_1},\ldots,K_{m_{2k}})$. For two vertices $u$ and $v$ to be strongly cospectral in $X$, Theorem \ref{sc}(3) implies that they both belong to $O_{m_j}$ if $j$ is odd or $K_{m_j}$ if $j$ is even. Hence, we have three cases.
\begin{itemize}
\item Let $j=1$. Then vertices $u$ and $v$ of $O_{m_1}$ are strongly cospectral in $X$ if and only if $m_1=2$.
\item Suppose $j\geq 3$ is odd so that vertices $u$ and $v$ both belong to $O_{m_j}$. Then we can write $X=(((Y \cup O_{m_j})\vee K_{m_{j+1}})\cup \ldots)\vee K_{m_{2k}}$ for some graph $Y$. Now, Theorem \ref{sc}(1) and an inductive argument implies that either (i) $u$ and $v$ are strongly cospectral in $(Y \cup X_{m_j})\vee X_{m_{j+1}}$ and $m_j\notin\sigma_{uv}^-(L(X_j))$ or (ii) $Y \cup X_{m_j}=O_2$. Since $Y$ has at least one vertex, the latter case does not hold. Hence, it must be that $u$ and $v$ are strongly cospectral in $(Y \cup O_{m_j})\vee K_{m_{j+1}}$. But this implies that $u$ and $v$ are strongly cospectral in $O_{m_j}$, which cannot happen since by Theorem \ref{sc}, the only graph with isolated vertices that admit strong cospectrality in the join is $O_2$, which is not equal to $Y \cup O_{m_j}$. 
\item Suppose $j\geq 3$ is even so that $u$ and $v$ belong to $K_{m_j}$. Then the same argument as the previous case implies that $u$ and $v$ are strongly cospectral in $Y\vee K_{m_{j}}$ and $m_j\notin\sigma_{uv}^-(L(X_j))$. But this implies that $u$ and $v$ are strongly cospectral in $K_{m_j}$, i.e., $m_j=2\in\sigma_{uv}^-(L(X_j))$, which is a contradiction.
\end{itemize}
Thus, in $X=\Gamma(O_{m_1},K_{m_1},\ldots,K_{m_{2k}})$, only the vertices of $O_{m_1}$ admit strong cospectrality in $X$, in which case $m_1=2$. The same holds for $X=\Gamma(O_{m_1},K_{m_1},\ldots,K_{m_{2k}})$. Cmbining this with Remark \ref{lastrem} completes the proof. 
\end{proof}

\section{Bounds}\label{secBounds}

From the previous sections, PST between $u$ and $w$ in $X$ need not guarantee PST between $u$ and $v$ in $X\vee Y$. The same can be said about periodicity of a vertex. Thus, we ask, if $u,v\in V(X)$ then how far can $|U_M(X\vee Y,t)_{u,v}|$ be from $|U_M(X,t)_{u,v}|$ as $t$ ranges over $\mathbb{R}$? We answer this question by providing an upper bound for the absolute value of $|U_M(X\vee Y,t)_{u,v}|-|U_M(X,t)_{u,v}|$. To do this, we define
\begin{equation}
\label{alpha}
\alpha_M(t)=\begin{cases}
   \dfrac{me^{-itn}+ne^{itm}-(m+n)}{m(m+n)},& \text{if } M=L\\
    \dfrac{e^{it\lambda^+}(k-\lambda^-)}{m\sqrt{D}}-\dfrac{e^{it\lambda^-}(k-\lambda^+)}{m\sqrt{D}}-\dfrac{e^{it k}}{m} ,& \text{if } M=A.
\end{cases}
\end{equation}
and
\begin{equation}
\label{beta}
T_M=\begin{cases}
   \{2j\pi/g:j\in\mathbb{Z}\},& \text{if } M=L\\
   \{t\in\mathbb{R}:e^{itk}=e^{it\lambda^+}=e^{it\lambda^-}\},& \text{if } M=A,
\end{cases}
\end{equation}
where $g=\operatorname{gcd}(m,n)$ and $\lambda^{\pm}$ are given in Lemma \ref{aqjoin}. If the context is clear, then we simply write as $\alpha_M(t)$ and $T_M$ as $\alpha(t)$ and $T$, respectively. If we add that $k,\ell\in\mathbb{Z}$ and $\Delta$ in Lemma \ref{aqjoin} is a perfect square, then we may write $T_A=\{2j\pi/h:j\in\mathbb{Z}\}$, where $h=\operatorname{gcd}(\lambda^{+}-k,\lambda^{-}-k)$. 

Combining (\ref{notoffL}) and (\ref{notoffA}) with (\ref{alpha}) yields the following lemma.

\begin{lemma}
\label{lemb}
We have  $\alpha_M(t)=0$ if and only if $t\in T_M$. Moreover, for any $u,v\in V(X)$ and for all $t$,
\begin{equation*}
U_L(X\vee Y,t)_{u,v}-e^{itn}U_L(X,t)_{u,v}=e^{itn}\alpha_L(t)\quad \text{and} \quad U_A(X\vee Y,t)_{u,v}-U_A(X,t)_{u,v}=\alpha_A(t).
\end{equation*}
\end{lemma}

\subsection{Laplacian case}

\begin{theorem}
\label{boundL}
For all $u,v\in V(X)$ and for all $t$
\begin{equation}
\label{notoffL1}
\big|\ U_L(X\vee Y,t)_{u,v}-e^{itn}U_L(X,t)_{u,v}\ \big|\leq 2/m
\end{equation}
with equality if and only if $\nu_2(m)=\nu_2(n)$, in which case equality holds in (\ref{notoffL1}) at time $\tau=j\pi/g$, where $j$ is any odd integer and $g=\operatorname{gcd}(m,n)$. Moreover, if equality holds in (\ref{notoffL1}), then
\begin{equation}
\label{notoffL2}
|\ U_L(X\vee Y,\tau)_{u,v}-e^{i\tau n}U_L(X,\tau)_{u,v}\ |=U_L(X\vee Y,\tau)_{u,v}+U_L(X,\tau)_{u,v}=2/m.
\end{equation}

\end{theorem}
\begin{proof}
Lemma \ref{lemb} yields (\ref{notoffL1}) with equality if and only if $e^{i\tau m}=e^{i\tau n}=-1$ for some $\tau>0$. Equivalently, $\nu_2(m)=\nu_2(n)$, in which case $\tau=j\pi/g$ for some odd $j$. From this, (\ref{notoffL2}) is immediate. 
\end{proof}

\begin{corollary}
\label{boundL1}
For all $u,v\in V(X)$ and for all $t$,
\begin{equation}
\label{notoffL3}
\bigg|\ |U_L(X\vee Y,t)_{u,v}|-|U_L(X,t)_{u,v}|\ \bigg|\leq 2/m
\end{equation}
with equality if and only if $\nu_2(m)=\nu_2(n)$, $U_L(X,\tau)_{u,v}\in\mathbb{R}$ and either $U_L(X,\tau)_{u,v}\leq 0$ or $U_L(X,\tau)_{u,v}\geq 2/m$, in which case equality holds in (\ref{notoffL3}) at time $\tau$ in Theorem \ref{boundL}.
\end{corollary}

\begin{proof}
Applying the triangle inequality to (\ref{notoffL1}) yields (\ref{notoffL3}). Using (\ref{notoffL2}), we get equality in (\ref{notoffL3}) if and only if $\big|\ |U_L(X\vee Y,\tau)_{u,v}|-|U_L(X,\tau)_{u,v}|\ \big|=\big|\ \big|\frac{2}{m}-U_L(X,\tau)_{u,v}\big |-|U_L(X,\tau)_{u,v}|\ \big|=\frac{2}{m}$. Result follows.
\end{proof}

\subsection{Adjacency case}

The following is an analogue of Theorem \ref{boundL} for the adjacency case.
\begin{theorem}
\label{boundA}
For all $u,v\in V(X)$ and for all $t$,
\begin{equation}
\label{notoffA1}
\big|\ U_A(X\vee Y,t)_{u,v}-U_A(X,t)_{u,v}\ \big|\leq 2/m
\end{equation}
with equality if and only if there is a time $\tau>0$ such that $e^{i\tau\lambda^+}=e^{i\tau\lambda^-}=-e^{i\tau k}$, in which case $
\alpha_A(t)=-2e^{i\tau k}/m$. If we add that $k,\ell\in\mathbb{Z}$ and $\Delta$ in Lemma \ref{aqjoin} is a perfect square, then the latter condition yields $\{k,\lambda^{\pm}\}\subseteq\mathbb{Z}$ and $\nu_2(\lambda^+-k)=\nu_2(\lambda^--k)$, in which case $\tau=j\pi/h$, where $h=\operatorname{gcd}(\lambda^+-k,\lambda^--k)$ and $j$ is any odd.
\end{theorem}

\begin{proof}
Since $\lambda^{\pm}=\frac{1}{2}(k+\ell\pm\sqrt{D})$, one checks that
\begin{equation}
\label{complex}
\begin{split}
\frac{e^{it\lambda^+}(k-\lambda^-)}{m\sqrt{D}}-\frac{e^{it\lambda^-}(k-\lambda^+)}{m\sqrt{D}}&=\frac{e^{it\lambda^+}(k-\ell+\sqrt{D})}{2m\sqrt{D}}-\frac{e^{it\lambda^-}(k-\ell-\sqrt{D})}{2m\sqrt{D}}\\
&=\frac{1}{m\sqrt{D}}\left[\sqrt{D}\cos(t\sqrt{D}/2)+i(k-\ell)\sin(t\sqrt{D}/2)\right].
\end{split}
\end{equation}
Consequently, the following equation holds for all $t$
\begin{equation}
\label{bound1}
\begin{split}
\left|\frac{e^{it\lambda^+}(k-\lambda^-)}{m\sqrt{D}}-\frac{e^{it\lambda^-}(k-\lambda^+)}{m\sqrt{D}}\right|
&=\frac{1}{m\sqrt{D}}\sqrt{D+((k-\ell)^2-D)\sin^2(t\sqrt{D}/2)}.
\end{split}
\end{equation}
Since $mn>0$ and $D=(k-\ell)^2+4mn$, we get that $(k-\ell)^2-D=-4mn<0$. Hence, (\ref{bound1}) gives us $\left|\frac{e^{it\lambda^+}(k-\lambda^-)}{m\sqrt{D}}-\frac{e^{it\lambda^-}(k-\lambda^+)}{m\sqrt{D}}\right|\leq\frac{1}{m}$. Combining this with (\ref{notoffA}) yields
\begin{center}
$\big|U_A(X\vee Y,t)_{u,v}-U_A(X,t)_{u,v}\big|\leq \left|\frac{e^{it\lambda^+}(k-\lambda^-)}{m\sqrt{D}}-\frac{e^{it\lambda^-}(k-\lambda^+)}{m\sqrt{D}}\right|+\frac{1}{m}\leq \frac{2}{m}$
\end{center}
which establishes (\ref{notoffA1}). Using the first equality in (\ref{complex}), we can write (\ref{notoffA}) as
\begin{center}
$U_A(X\vee Y,t)_{u,v}-U_A(X,t)_{u,v}=\frac{e^{it\lambda^+}(k-\ell+\sqrt{D})}{2m\sqrt{D}}-\frac{e^{it\lambda^-}(k-\ell-\sqrt{D})}{2m\sqrt{D}}-\frac{e^{it k}}{m}$.
\end{center}
Thus, equality holds in (\ref{notoffA1}) if and only if there is a time $\tau$ such that
\begin{equation}
\label{exps}
e^{i\tau\lambda^+}=e^{i\tau\lambda^-}=-e^{i\tau k},
\end{equation}
in which case, $U_A(X\vee Y,t)_{u,v}-U_A(X,t)_{u,v}=-\frac{2e^{itk}}{m}$. This implies that $\{k,\lambda^{\pm}\}$ satisfies the ratio condition. Thus, if $k,\ell\in\mathbb{Z}$ and $D$ is a perfect square, then Lemma \ref{per} yields $\{k,\lambda^{\pm}\}\subseteq\mathbb{Z}$. Hence, (\ref{exps}) holds if and only if $\nu_2(\lambda^+-k)=\nu_2(\lambda^--k)$, in which case $\tau=j\pi/g$ for any odd integer $j$. 
\end{proof}

\begin{corollary}
\label{boundA1}
For all $u,v\in V(X)$ and for all $t$,
\begin{equation}
\label{notoffA2}
\begin{split}
\bigg|\ |U_A(X\vee Y,t)_{u,v}|-|U_A(X,t)_{u,v}|\ \bigg|&\leq 2/m
\end{split}
\end{equation}
with equality if and only if the bound in (\ref{notoffA1}) holds at time $\tau$ in Theorem \ref{boundA}, $U_A(X,\tau)_{u,v}=|U_A(X,\tau)_{u,v}|e^{i\tau k}$ and $|U_A(X,\tau)_{u,v}|\geq 2/m$.
\end{corollary}

\begin{proof}
Again applying triangle inequality to (\ref{notoffA1}) yields (\ref{notoffA2}). Making use of Theorem \ref{boundA}, we get equality in (\ref{notoffA2}) if and only if $\big|\ |U_A(X,\tau)_{u,v}-2e^{i\tau k}/m\ |-|U_A(X,\tau)_{u,v}|\ \big|=2/m$. Thus, our result is immediate.
\end{proof}

Combining Corollary \ref{boundL1} and \ref{boundA1} yields the following result.

\begin{corollary}
\label{mimic}
Let $M\in\{A,L\}$. For all $u,v\in V(X)$ and for all $t$,
\begin{equation*}
\bigg|\ |U_M(X\vee Y,t)_{u,v}|-|U_M(X,t)_{u,v}|\ \bigg|\rightarrow 0\quad \text{as $m\rightarrow\infty$}.
\end{equation*}
\end{corollary}

From Corollary \ref{mimic}, we find that for vertices in $X$, 
the state transfer properties  of the quantum walk on $X\vee Y$ tend to mimic those of the quantum walk on $X,$ as we increase the number of vertices of $X$.

\subsection{Tightness} 

Define $F(t)_{u,v}:=|U_M(X\vee Y,t)_{u,v}|-|U_M(X,t)_{u,v}|$. From Lemma \ref{lemb}, $F(t)_{u,v}=\big|U_M(X,t)_{u,v}+\alpha_M(t)\big|-|U_M(X,t)_{u,v}|$.
for any $u,v\in V(X)$. We now illustrate that the values of $F(t)_{u,v}$ can be zero, positive or negative.

\begin{lemma}
\label{aa}
 The following are immediate from (\ref{beta}).
\begin{enumerate}
\item If $\tau\in T_M$ then $\alpha_M(\tau)=0$, and so $F(\tau)_{u,v}=0$.
\item If $U_M(X,\tau)_{u,v}=0$ for some $\tau\in\mathbb{R}\backslash T_M$, then $F(\tau)_{u,v}=|\alpha_M(\tau)|>0$. If $X$ is disconnected and $u$ and $v$ are in different connected components, then $U_M(X,t)_{u,v}=0$ for all $t\in\mathbb{R}$, and so $F(t)_{u,v}>0$ for all $t\in\mathbb{R}\backslash T_M$.
\item Let $|U_M(X,\tau)_{u,v}|=1$ for some $\tau\in\mathbb{R}\backslash T_M$. Since $\alpha_M(t)\neq 0$ and $|U_M(X\vee Y,t)_{u,v}|\leq 1$, we get $F(\tau)_{u,v}<0$. If $u$ is isolated in $X$, then $U_M(X,t)_{u,u}=1$ for all $t$, and so $F(t)_{u,u}< 0$ for all $t\in \mathbb{R}\backslash T_M$.
\end{enumerate}
\end{lemma}

We now use Lemma \ref{aa} to show that under mild conditions on $m$ and $n$, if $X$ has PST between $u$ and $v$, then there are choices of $t$ for which $|F(t)_{u,v}|=2/m$. Thus, the bounds in Corollaries \ref{boundL1} and \ref{boundA1} are tight.

\begin{example}
\label{exa}
Suppose $X$ exhibits Laplacian PST between $u$ and $v$ at $\pi/2$. Let $\tau=r\pi$ for any odd $r$ so that $U_L(X,\tau/2)_{u,v}=1$ and $U_L(X,\tau/2)_{u,u}=0$. If $\nu_2(m)=\nu_2(n)=1$, then we can write $T_L=\{j\pi/g':j\in\mathbb{Z}\}$, where $g'=g/2$ is odd. Hence, $\tau/2\in\mathbb{R}\backslash T_L$. In this case, $\alpha_L(\tau/2)=-\frac{2}{m}$, and so $F(\tau/2)_{u,u}=-F(\tau/2)_{u,v}=\frac{2}{m}$.
\end{example}


\begin{example}
\label{exaa}
Suppose $X$ admits adjacency PST between $u$ and $v$ at $\pi/2$. Let $k,\ell\in\mathbb{Z}$ and $\Delta$ be a perfect square so that, $k,\lambda^{+}$ and $\lambda^{-}$ are integers. Let $\tau=r\pi$ for any odd $r$ so that $U_A(X,\tau/2)_{u,v}=e^{i\tau k/2}$ and $U_A(X,\tau/2)_{u,u}=0$. If $\nu_2(\lambda^+-k)=1$ and $\nu_2(\lambda^--k)=1$, then we can write $T_A=\{j\pi/h':j\in\mathbb{Z}\}$ where $h'=h/2$ is odd. Thus, $\tau/2\in\mathbb{R}\backslash T_A$, and so $\alpha_A(\tau/2)=-\frac{2e^{i\tau k/2}}{m}$, which yields $F(\tau/2)_{u,u}=-F(\tau/2)_{u,v}=\frac{2}{m}$.
\end{example}

One checks that disconnected double cones with PST at $\frac{\pi}{2}$ satisfy the conditions in Examples \ref{exa} and  \ref{exaa}, and so they provide infinite families of graphs whereby the upper bound in Corollaries \ref{boundL1} and \ref{boundA1} is met. 

%
%

\section{Future work}
\label{secFw}

In this paper, we investigated strong cospectrality, periodicity and perfect state transfer in join graphs, and characterized the conditions in which these properties are preserved or induced in a join. In order to inspire further systematic study of quantum walks on join graphs, we present some open questions.

If $X$ is a weighted join graph on an odd number of vertices and $\phi(L(X),t)\in\mathbb{Z}[t]$, then Corollary \ref{pstLc3} implies that Laplacian PST does not occur in $X$. This is also known to be true for unweighted graphs on $n\in\{3,5\}$ vertices, which are not necessarily joins. This leads us to conjecture the following.

\begin{conjecture}
If $X$ is a simple connected unweighted graph on an odd number of vertices, then Laplacian perfect state transfer does not occur in $X$.
\end{conjecture}

We note however that if $X$ is a weighted join graph with an odd number of vertices, then it is possible to get PST in $X$ if we drop the condition that $\phi(L(X),t)\in\mathbb{Z}[t]$. Indeed, if $X$ is a weighted $K_2$ with vertices $u$ and $v$ joined by an edge of weight $\frac{a}{b}>0$ where $a>0$ is odd and $b\equiv 0$ (mod 4), then $\phi(L(X),t)\notin\mathbb{Z}[t]$, but one checks that $X\vee O_1$ is a weighted $K_3$ with PST between $u$ and $v$ at time $\frac{2\pi}{3}$.

Next, observe that $X$ is an induced subgraph of $X\vee Y$. In Corollaries \ref{boundL1} and \ref{boundA1}, we showed for $M\in\{A,L\}$ that $\big||U_M(X\vee Y,t)_{u,v}|-|U_{M}(X,t)_{u,v}|\big|\leq f(|V(X)|)$ for all $t$, where $f(|V(X)|)=2/|V(X)|$, and this inequality is tight for infinite families. In general, if $X$ is a subgraph of $Z$ induced by $W\subseteq V(Z)$ with $|W|\geq 2$, then for $u,v\in V(X)$, we ask: what is an upper bound for $\big||U_M(Z,t)_{u,v}|-|U_{M}(X,t)_{u,v}|\big|$ in terms of $|W|$? In particular, if $X$ is either complete or empty, what is $\sup_{t>0}\big||U_M(Z,t)_{u,v}|-|U_{M}(X,t)_{u,v}|\big|$? In line with Corollary \ref{mimic}, we further ask: for which induced subgraphs $X$ of $Z$ does $\big||U_M(Z,t)_{u,v}|-|U_{M}(X,t)_{u,v}|\big|\rightarrow 0$ for all $t$ as $|V(X)|\rightarrow\infty$?

If $X=K_m$ and $Y$ is a graph on $m$ vertices, then we have proof for $M\in\{A,L\}$ that $|U_M(X\vee Y,t)_{u,v}|-|U_{M}(X,t)_{u,v}|=0$ if and only if $t$ belongs to some countable set. Hence, we pose the question, are there families of graphs such that for some set $T$ with positive measure, $|U_M(X\vee Y,t)_{u,v}|-|U_{M}(X,t)_{u,v}|=0$ for all $t\in T$?

Finally, one of the surprising results in this paper is that the join operation can induce PST in the resulting graph under mild conditions. Owing to the fact that adjacency PST is rare in unweighted graphs \cite[Corollary 10.2]{Godsil2012a}, it would be interesting to determine other graph operations that induce PST in the resulting graph. In \cite[Theorem 12]{bhattacharjya2023quantum}, the authors show that the blow-up operation exhibits this property for some families of graphs (such as the blow-up of two copies of $K_{1,n}$). However, under Cartesian products, PST between vertices $(u,v)$ and $(w,x)$ with $(u,v)\neq (w,x)$ requires that either (i) $u$ and $w$, and $v$ and $x$ admit PST in the underlying graphs (ii) $u$ and $w$ admit PST and $v=x$ is periodic in the underlying graphs or (iii) $u=w$ is periodic and $v$ and $w$ admit PST in the underlying graphs. Hence, the Cartesian product cannot induce PST in the resulting graph.

\section*{Acknowledgements}
H. Monterde is supported by the University of Manitoba Faculty of Science and Faculty of Graduate Studies. S. Kirkland is supported by NSERC Discovery Grant RGPIN-2019-05408.

\bibliographystyle{alpha}
\bibliography{mybibfile}
\end{document}